%% file: main.tex
\documentclass[11pt]{article}
\usepackage{graphicx}

\input{header}

\title{Disjoint Paths in Expanders in Deterministic Almost-Linear Time \\via Hypergraph Perfect Matching}
\author{
Matija Buci\'{c}\thanks{Supported by NSF Grant DMS-2349013.}\\
\parbox[t]{5cm}{\centering 
University of Vienna\\ Princeton University}
\and
Zhongtian He\\
Princeton University
\and
Shang-En Huang\thanks{Supported by NSTC Grant No. 114-2222-E-002-004-MY3.}\\
National Taiwan University
\and
Thatchaphol Saranurak\thanks{Supported by NSF Grant CCF-2238138.}\\
University of Michigan
}
\date{}

\begin{document}

\begin{titlepage}
    \thispagestyle{empty}
    \maketitle
    \begin{abstract}
        \thispagestyle{empty}
        We design efficient deterministic algorithms for finding short edge-disjoint paths in expanders. Specifically, given an $n$-vertex $m$-edge expander $G$ of conductance $\phi$ and minimum degree $\delta$, and a set of pairs $\{(s_i,t_i)\}_i$ such that each vertex appears in at most $k$ pairs, our algorithm deterministically computes a set of edge-disjoint paths from $s_i$ to $t_i$, one for every $i$,
\begin{enumerate}
    \item each of length at most $18 \log (n)/\phi$ and in $mn^{1+o(1)}\min\{k, \phi^{-1}\}$ total time, assuming \\  $\phi^3\delta\ge (35\log n)^3 k$, or
    \item each of length at most $n^{o(1)}/\phi$ and in total $m^{1+o(1)}$ time, assuming   $\phi^3 \delta \ge n^{o(1)} k$.
\end{enumerate}
Before our work, deterministic polynomial-time algorithms were known only for expanders with constant conductance and were significantly slower.

To obtain our result, we give an almost-linear time algorithm for \emph{hypergraph perfect matching} under generalizations of Hall-type conditions~\cite{Haxell1995}, a powerful framework with applications in various settings, which until now has only admitted large polynomial-time algorithms~\cite{Annamalai18}. %

    \end{abstract}
\end{titlepage}

\setcounter{tocdepth}{3}
\setcounter{secnumdepth}{3}

\tableofcontents
\thispagestyle{empty}
\newpage
\addtocontents{toc}{\protect\thispagestyle{empty}} 
\setcounter{page}{1}

\section{Introduction}

We study the classical \emph{edge-disjoint paths problem}: given an undirected graph \( G = (V, E) \) with \( n \) vertices and \( m \) edges, and a collection of \( q \) source-destination pairs \( \{(s_i, t_i)\}_{i=1}^q \), the goal is to find \( q \) pairwise edge-disjoint paths such that the \( i \)-th path connects \( s_i \) to \( t_i \). This problem arises naturally in network routing, VLSI design, and the theory of multicommodity flows, where independent communication requests must be served without interference. 

In general undirected graphs, the problem is NP-complete~\cite{Karp}, but admits almost-linear time algorithms for constant $q$ \cite{RS95,korhonen2024minor}.  
For directed graphs, the problem is NP-complete~\cite{Fortune80} even when \( q = 2 \). This motivates the development of efficient algorithms for large $q$ on special classes of graphs, such as \emph{expanders}, i.e., graphs with large conductance. Recall that the conductance of a vertex set $S$ is $\Phi_G(S)=|\partial_G(S)|/\min\{\vol(S), \vol(\overline{S})\} \in [0,1]$ where $\partial_G(S):=E(S,V\setminus S)$ and \(\vol(S) := \sum_{v \in S} \deg(v)\), and the conductance of graph $G$ is $\Phi_G = \min_{\emptyset \neq S \subsetneq V} \Phi_G(S)$.  Intuitively, the larger the conductance is, the larger the cuts are required to be in order to separate a large volume of the graph.

\paragraph{Disjoint Paths on Strong Expanders.}
Prior works focused on regular expanders with strong assumptions.
In particular, the class of expander graphs previously considered required 
not only to have constant conductance 
but also to have strictly more-than-$0.5$ or even close-to-$1$ conductance for small subsets.
Peleg and Upfal~\cite{peleg1989local} showed that edge-disjoint paths exist for \(q = \Theta(n^\epsilon)\) demand pairs \((s_i, t_i)\), where the endpoints are pairwise disjoint and \(\epsilon > 0\) depends on the expansion of \(G\). Subsequent works~(e.g. \cite{BroderFriezeUpfal94,BroderFriezeUpfal99,LeightonRao99,LeightonLuRaoSrinivasan00}) improved the bound on \(q\), culminating in a result of Frieze~\cite{frieze2000disjoint}, who showed that disjoint paths exist and can be found in randomized polynomial time for \(q = \Theta(n/\log n)\), which is optimal. Alon and Capalbo~\cite{alon2007online} presented a deterministic polynomial-time algorithm that also applies in an \emph{online} setting, where the demand pairs arrive one by one and each path must be routed disjointly from those previously constructed. More recently, Draganić and Nenadov~\cite{draganic2023finding} extended this to the more challenging \emph{online with deletions} model and provided a refined runtime analysis, achieving a deterministic algorithm with running time \(O(m^3)\) per demand pair. %

Unfortunately, all the above algorithms fail on expanders with sub-constant conductance.
In the context of fast graph algorithms, however, many algorithms must work with expanders whose conductance is at most $1/\log(n)$ because this is the inherent limit of the popular expander decomposition framework \cite{alev2018graph}. This motivates algorithms that work with a weaker expansion guarantee.

\paragraph{Our Results.}
In this work, we design fast algorithms for the disjoint-paths problem under this weaker expansion regime and simultaneously significantly improve the running time compared to prior approaches.

\begin{theorem}
\label{thm: routing conductance}
    Let $G$ be an $n$-vertex graph with $m$ edges, conductance $\phi,$ and minimum degree $\delta$. Given a set of demand pairs $\{(s_i,t_i)\}_i$ such that each vertex appears in at most $k$ pairs, there exists a deterministic algorithm that computes a set of edge-disjoint paths from $s_i$ to $t_i,$ one for every $i$,
    \begin{enumerate}
        \item each of length at most $18\log(n)/\phi$ and in $mn^{1+o(1)}\min\{k, \phi^{-1}\}$ time, assuming \\
        $\phi^3\delta\ge (35\log n)^3 k$, or
        \item each of length at most $n^{o(1)}/\phi$ and in $m^{1+o(1)}$ time, assuming 
        $\phi^3 \delta \ge n^{o(1)} k$.
\end{enumerate}
\end{theorem}
We note that, if randomness is allowed, the problem becomes significantly easier: we can split the expander by randomly partitioning the edges, and applying an approximate disjoint paths algorithm to each part yields a solution using known techniques.

An immediate application of our disjoint paths theorem is the first almost-linear time deterministic algorithm for splitting expanders into many sparser expanders.

\begin{corollary}
\label{cor: splitting expanders}    
Let $G$ be an $n$-vertex graph with $m$ edges, conductance $\phi,$ and minimum degree $\delta$. There exists a deterministic algorithm that partitions $G$ into $k$ disjoint subgraphs, i.e.\ $G_i = (V,E_i)$ and $\bigsqcup_{i=1}^k E_i\subseteq E$, such that
    \begin{enumerate}
        \item each $G_i$ has conductance  $\Omega(\phi/\log n)$, in $mn^{1+o(1)}\min\{k, \phi^{-1}\}$ time, assuming \\
        $\phi^3\delta\ge (73\log n)^3 k$, or
        \item each $G_i$ has conductance $\Omega(\phi/n^{o(1)}),$ in $m^{1+o(1)}$ time, assuming 
        $\phi^3 \delta \ge n^{o(1)} k$.
\end{enumerate}
\end{corollary}

The proof of \Cref{cor: splitting expanders} is deferred to \Cref{app: splitting expanders}. Before our result, Frieze and Molloy \cite{frieze1999splitting} showed a deterministic algorithm using the Lovász Local Lemma, which avoids incurring an extra polylogarithmic factor; however, it requires a large polynomial running time. 
If randomness is allowed, a simple random partition suffices as shown in Section~8 of~\cite{wulff2017fully}.

\paragraph{Reduction to Hypergraph Matching.}
Our proof of \Cref{thm: routing conductance} is based on a reduction to the hypergraph matching problem. Below, we will explain this problem, the reduction, and why the prior tools are still too slow for us. We start with the definitions.%

\begin{definition}[Bipartite Hypergraph Matching]\label{def:hyper matching}
    A \emph{bipartite hypergraph} $H = (A,B,E)$ is a hypergraph with two disjoint vertex sets $A$ and $B$ where $|e\cap A| = 1 $ and $|e \cap B| \ge 1$ for every edge $e\in E$. We say $H$ is \emph{$r$-bounded} if $|e\cap B| \le r$ for every $e \in E$.
    A \emph{matching} $M \subseteq E$ is a set of vertex-disjoint edges. We say $M$ is \emph{perfect} if it contains every vertex from $A$. 
\end{definition}

We now explain the reduction: given a graph $G = (V,E_G)$ and a set of demand pairs $\{s_i,t_i\}_{i=1}^q$, we construct a hypergraph $H = (A,B,E_H)$ with $|A| = q,$ so that each vertex in $A$ stands for one demand pair, $|B| = m$ so that each vertex in $B$ stands for an edge in $E_G$, each hyperedge contains one vertex in $A$ corresponding to a demand $i\in [q]$ and a set of vertices in $B$ corresponding to a set of edges in $E_G$, that form a path from $s_i$ to $t_i$. Thus, a perfect matching in $H$ corresponds to a set of disjoint paths connecting all demand pairs.

However, finding maximum or perfect matchings in hypergraphs is NP‑complete even for \linebreak 2‑bounded bipartite hypergraphs, via a straightforward reduction from \(\textsc{3-Dimensional Matching}\). But, under \emph{Haxell's condition} \cite{Haxell1995}, which is a natural generalization of Hall’s condition \cite{Hall1935}, the problem becomes tractable. For every subset \(S \subseteq A\), define
\[
E_S \;:=\; \{e \in E_H : e \cap S \neq \emptyset\}, 
\quad
\tau(E_S) \;:=\; \min\bigl\{|T| : T \subseteq B,\; \forall e \in E_S,\; e \cap T \neq \emptyset\bigr\}.
\]
That is, $\tau(E_S)$ is the minimum number of $B$-vertices required to hit all hyperedges incident to $S$. 
Haxell’s condition requires that \(\tau(E_S)\) be proportional to \(|S|\); more precisely, assuming every hyperedge size is bounded by \(r\),  we have for every \(S \subseteq A\)
\[
\tau(E_S) > (2r - 1)(|S| - 1). 
\]
Under this condition, Haxell~\cite{Haxell1995} proved the existence of a perfect matching. Later, Annamalai~\cite{Annamalai18} gave a polynomial-time algorithm for finding perfect matchings under a slightly stronger version of Haxell’s condition in $r$-bounded bipartite hypergraphs, where the term $(2r - 1)$ is replaced by $(2r - 1 + \epsilon)$. However, this algorithm incurs an exponential dependence on both the uniformity parameter \(r\) and the slack parameter $\epsilon$. This is far too slow for our purposes.

\paragraph{Our Technical Contribution.}
To carry out this reduction efficiently for the disjoint path problems, our contribution is twofold. 
Firstly, we show that Annamalai's exponential dependency on $r$ and $\epsilon$ is unnecessary given stronger Haxell's condition, defined as follows.

\begin{definition}[Strong Haxell’s condition]\label{def:haxell}
We say that a bipartite hypergraph \(H=(A,B,E)\) satisfies the \emph{\(\varphi\)‑strong Haxell condition} if, for every subset \(S\subseteq A\),
\[
\tau(E_S)\;\ge\;\varphi|S|,
\]
where \(E_S = \{e\in E : e\cap S \neq \emptyset\}\).
\end{definition}

Below, we present an almost-linear %
time hypergraph matching algorithm when $\varphi = \omega(r^2)$.

\begin{theorem}
\label{thm: almost linear time hypergraph perfect matching}
    Let $H=(A, B, E)$ be an $r$-bounded bipartite hypergraph that satisfies the $\varphi$-strong Haxell condition with $\varphi \ge d(n)r^2$ 
    for some parameter $d(n)\ge 4$.
    Then, there exists an algorithm which runs in $O(p^{1+1/\Omega(\sqrt{\log d(n)})})$
    time that computes a perfect matching, where $p := \sum_{e\in E} |e|$.
\end{theorem}

\Cref{thm: almost linear time hypergraph perfect matching} effectively brings the powerful tool of hypergraph perfect matching into the realm of almost-linear-time graph algorithms. See \Cref{tab:hyper matching} for prior results.
Given its prior applications discussed in \Cref{sec:related}, we believe it may find further applications.

\begin{table}[H]
\centering
\small{

\begin{tabular}{|c|c|c|}
\hline 
Reference & Expansion condition $\gamma$ & Runtime\tabularnewline
\hline 
\hline 
Haxell \cite{Haxell1995} %
& $\tau(E_{S})>(2r\text{\textminus}1)(|S|\text{\textminus}1)$ & existential \tabularnewline
\hline 
Annamalai \cite{Annamalai18} %
& $\tau(E_{S})>(2r\text{\textminus}1+\epsilon)(|S|\text{\textminus}1)$ & $p^{\poly(r/\epsilon)}$\tabularnewline
\hline Annamalai, Kalaitzis, Svensson
\cite{annamalai2017combinatorial} & $\tau(E_{S})>(c_{0}r/\epsilon)(|S|\text{\textminus}1)$ & $\poly(p)^\star$ \tabularnewline
\hline 
\textbf{\Cref{thm: almost linear time hypergraph perfect matching}} & $\tau(E_{S})>r^{2}d(n)|S|$ & $O(p^{1+1/\Omega(\sqrt{\log d(n)})})$\tabularnewline
\hline 
\end{tabular}
}
\caption{Hypergraph perfect matching algorithms in literature. $(\star)$ The algorithm of \cite{annamalai2017combinatorial} only computes an \emph{$\epsilon$-near‐perfect} matching $M$. That is, $M$  becomes a perfect matching only after discarding an $\epsilon$‐fraction of each hyperedge.}\label{tab:hyper matching}
\end{table}

Secondly, observe that the reduction to hypergraph matching discussed below \Cref{def:hyper matching} takes super-polynomial time simply because $H$ is too big.
Even if we include only the paths of length at most $r = \poly\log n$, the size of $E_H$ can be quasi-polynomial. Thus, even if we use the almost-linear time algorithm from \Cref{thm: almost linear time hypergraph perfect matching} on $H$, the running time would still be super-polynomial.

To bypass this, we instead access $E_H$ only through the edge oracles. These oracles must be efficiently implemented yet strong enough to support the hypergraph perfect matching algorithms. 
Our final algorithm cannot use \Cref{thm: almost linear time hypergraph perfect matching} as a black box but will need its strengthened version, which accesses $H$ only via the edge oracles. We formalize this notion of edge oracles as \emph{half layer} oracles (defined formally in \Cref{sec: half layer}).
Intuitively speaking, a half layer is a collection of hyperedges which are vertex disjoint within the $B$ part, with some additional properties.
We gave two implementations of the edge oracles for the hypergraphs described above arising from our disjoint paths problem:
a simple one that leads to the $mn^{1+o(1)}\min\{k, \phi^{-1}\}$ final running time, and a more advanced one based on recently developed multi-commodity flow algorithms \cite{haeupler2024low} that leads to the almost-linear $m^{1+o(1)}$ running time.

\subsection{Related Works}
\label{sec:related}

\paragraph{Applications of Disjoint Paths to Classical Combinatorial Problems.}
Finding edge disjoint paths by means of Haxell's condition has proved to be very powerful in attacking several classical combinatorial problems in recent years including Erd\H{o}s-Gallai cycle decomposition conjecture \cite{BM23}, Katona's path separation problem \cite{letzter2024separating}, finding regular subgraphs \cite{chakraborti2025edge}, and Hamiltonicity of weak expanders \cite{letzter2025nearly}. Our efficient algorithm for the disjoint paths problem paves the way for efficient versions of all of these results.

\paragraph{Approximation Algorithms for Disjoint Paths in Expanders.}

A different line of work tries to devise algorithms that only find disjoint paths between a \emph{subset} of demand pairs.
The approximation ratio is defined as the ratio between the maximum number of demand pairs that can be connected via disjoint paths and the number of demands satisfied by the algorithm.
Chuzhoy, Kim, and Nimavat \cite{ChuzhoyKN18} show that there is no polynomial time algorithm for achieving a $2^{\Omega(\log^{1-\epsilon} n)}$-approximation for any constant $\epsilon > 0$, under a reasonable assumption that $\mathsf{NP}\not\subseteq \mathsf{RTIME}(n^{\mathrm{polylog}(n)})$.

However, on expander graphs,
Aumann and Rabani~\cite{AumannRabani1998} gave a deterministic offline $O(\log n)$‐approximation algorithm, and Kleinberg and Rubinfeld~\cite{KleinbergRubinfeld2006} presented a deterministic online greedy algorithm achieving an $O(\log n\,\log\log n)$ guarantee. Subsequent randomized approaches further improved these bounds: Kolman and Scheideler~\cite{KolmanScheideler2002} obtained an $O(\log n)$‐approximation, and Chakrabarti et al.~\cite{ChakrabartiGuptaKumar2007} achieved an $O(\sqrt{\log n})$‐approximation in strong expanders.
\thatchaphol{What about fast algorithm?}

For designing fast algorithms, 
there has been a line of research work 
on approximate multi-commodity flow with short flow path \cite{HaeuplerHS23, haeupler2024low, HaeuplerH025} 
and distributed routing \cite{chang2024deterministic, ChangS20, GhaffariL18, GhaffariKS17}.
On expanders, these algorithms achieves $n^{o(1)}$-approximation for the disjoint paths problem in almost-linear time.

\paragraph{Applications of Hypergraph Matching in Fair Allocation.} 
Hypergraph perfect matching has found  applications in the restricted max–min fair allocation problem (also known as the Santa Claus problem), which aims to allocate indivisible resources to players in a balanced manner. In the influential work of Asadpour, Feige, and Saberi~\cite{AsadpourFeigeSaberi2012}, the problem is reduced to finding perfect matchings in a bipartite hypergraph defined by the configuration LP, and a local search procedure inspired by Haxell’s theorem is used to establish a constant-factor integrality gap. Annamalai et al.~\cite{annamalai2017combinatorial} gave a combinatorial algorithm that finds such matchings efficiently, yielding a 13-approximation. Davies et al.~\cite{DaviesRothvossZhang2020} improved the approximation to $4+\varepsilon$ using a matroid-based extension of the hypergraph matching framework. Further generalizations to non-uniform hypergraphs were obtained by Bamas et al.~\cite{BamasGargRohwedder2020}. These results rely on structural insights into hypergraph matchings and build upon techniques developed for Haxell-type conditions.

\subsection{Organization and Preliminaries}

\paragraph{Organization.}
We give a high-level overview of our techniques in \Cref{sec:technical overview}. 
In \Cref{sec: half layer}, we define an abstraction called a \emph{half layer} and state that computing a hypergraph perfect matching can be reduced to the edge oracle that can find a (maximal or approximate) half layer. 
In \Cref{sec: disjoint paths}, specific to the application of finding disjoint paths on expanders, we show how to efficiently find the half layer without explicitly constructing the hypergraph. 
By applying the reduction from \Cref{sec: half layer}, we will obtain our main result (\Cref{thm: routing conductance}).
Finally, in \Cref{sec:proof hyper matching}, we give the proof of the reductions to finding a half layer from \Cref{sec: half layer}. Our arguments in this section build upon the analysis framework introduced by Annamalai~\cite{Annamalai18}, with several key changes.

\paragraph{Preliminaries.} Throughout the paper all our logarithms are in base 2.
Given a bipartite hypergraph $H=(A,B,E)$ we will denote its number of vertices by $n(H) := |H| := |A|+|B|$ and the total volume of its edges by $p(H):= \sum_{e\in E} |e|$. When the hypergraph $H$ is clear from context we often omit it and write simply $n$ or $p$. Given a subset of edges $X \subseteq E$, we write $A(X)=\bigcup_{e \in X} e \cap A$ and similarly $B(X)=\bigcup_{e \in X} e \cap B$. We define the \emph{rank} of an edge $e$ to be $|e\cap B|$.

\section{Technical Overview}\label{sec:technical overview}
In this section, we explain high-level ideas behind our two main technical contributions as discussed in the introduction.

\paragraph{Hypergraph Perfect Matching in Almost Linear Time. }

First, we describe the idea of our almost-linear-time algorithm for hypergraph bipartite perfect matching using stronger Haxell's condition. This part involves reanalyzing Annamalai's framework \cite{Annamalai18}, but we can also simplify a significant part of his argument, given that we are working with an even stronger Haxell's condition.

At a high level, our algorithm can be seen as a hypergraph generalization of Motwani's algorithm \cite{motwani1989expanding} for finding a perfect matching in normal bipartite graphs satisfying the \emph{stronger Hall's condition.} This is a bipartite graph $G=(A,B,E)$ such that for all $S\subseteq A$, $|N(S)|\ge(1+\epsilon)|S|$ where $\epsilon>0$ and $N(S)$ is the neighbor set of $S$. Motwani showed that, for any non-perfect matching $M$, i.e., some vertex from $A$ remains unmatched, there must exist an alternating path for $M$ of length only $O(\frac{\log n}{\epsilon})$ that can augment the size of $M$. Thus, just by ensuring that there is no short alternating path of length $O(\frac{\log n}{\epsilon})$, we can conclude that $M$ is perfect. In a normal graph, this can be done by making $O(\frac{\log n}{\epsilon})$ calls to blocking flows, leading to an algorithm with $O(\frac{m\log n}{\epsilon})$ time for finding a perfect matching. In contrast, if $G$ only satisfies the standard Hall's condition, then the alternating path might have length $\Omega(n)$. The upshot is that, the stronger expansion guarantee is, the easier it is to find a perfect matching. 

The stronger Haxell's condition from \Cref{def:haxell} is precisely the hypergraph generalization of the stronger Hall's condition. Haxell \cite{Haxell1995} showed, once the expansion is strong enough, i.e., $\varphi>(2r-3)$, the perfect matching must exist in any $r$-bounded bipartite hypergraph. By slightly strengthening the condition to $\varphi>(2r-3+\epsilon)$, Annamalai \cite{Annamalai18} further showed that one can find it in $p^{\poly(r/\epsilon)}$ time. In hypergraphs, the notion of alternating paths become \emph{alternating trees} (formally defined in \Cref{sec:proof hyper matching}). The running time of the hypergraph perfect matching algorithm depends \emph{exponentially} on the depth of this tree, in contrast to the linear dependency in the case of  normal graphs. 
In Annamalai's algorithm \cite{Annamalai18}, the depth of his alternating trees can be as large as $\Omega(\log(n)\cdot \poly\log (r))$, leading to super-polynomial time when $r = \omega(1)$.

Our idea is to observe that an even stronger Haxell's condition when $\varphi\ge d(n)r^{2}$ naturally improves the depth of the alternating trees to only $O\left(\frac{\log n}{\log d(n)}\right)=o(\log n)$ if we set $d(n)=\omega(1)$.
Given the small depth, the only remaining algorithmic component is how to ``grow an alternating tree''. Our algorithm will grow an alternating tree layer by layer. At the end, the running time is of the form $T\cdot2^{O(\log(n) / \sqrt{\log d(n)})} =Tn^{o(1)}$ where $T$ is the time required to grow the alternating tree by one layer. If the hypergraph with total volume $p$ is given to us explicitly, then a simple greedy algorithm gives $T=O(p)$. This immediately leads to a proof of \Cref{thm: almost linear time hypergraph perfect matching} and will be formally proved in \Cref{sec: half layer}. 

However, in our disjoint-path applications, we cannot explicitly construct the underlying hypergraph fast enough, which leads us to the second challenge.

\paragraph{Growing Alternating Tree for Disjoint Paths without Explicit Hypergraph.}

Next, we describe how to grow a layer of the alternating tree without explicitly constructing the hypergraph described below \Cref{def:hyper matching}. We consider this our primary technical contribution.

To keep the presentation modular, we formally introduce a new abstraction called \emph{half-layer} in \Cref{sec: half layer}. In contrast to the alternating tree, this object does not require background knowledge of Annamalai's framework. Roughly speaking, finding a half layer corresponds to growing half a layer of an alternating tree from the vertex set $A$ to the set $B$. Another half a layer from $B$ to $A$ is relatively easy to build. 

In fact, we introduce two variants of half-layers: a \emph{maximal} half-layer and an \emph{approximate} half-layer, defined in \Cref{def: half layer,def:apx half layer}, respectively. For intuition, on normal graphs, finding a maximal half-layer corresponds to growing a layer of the breadth-first-search tree on the residual graph for finding an alternating path. The caveat is that, for efficiency, we also bound the maximum degree in the tree to be at most $d(n)$. On the other hand, an approximate half-layer has no direct analogy in normal graphs because it has bi-criteria approximation. But for intuition, if we ignored the approximation with respect to the rank of hyperedges, finding an approximate half-layer would correspond to, in normal graphs, finding an approximately largest (but possibly not maximal) set of edges to grow a layer in the breadth-first-search tree on the residual graph. 

In the next section, \Cref{thm: quadratic time hypergraph perfect matching,lem: fast perfect matching} formally state that given an efficient algorithm for finding a maximal or an approximate half layer with running time $T$, we can obtain a hypergraph perfect matching with running time $Tn^{o(1)}$. The formal proofs of these statements are given in \Cref{sec:proof hyper matching} and follow from our first technical contribution on the almost-linear time hypergraph perfect matching algorithm, and by showing that, indeed, our half-layer abstraction does fit into Annamalai's framework.

\section{Reducing Hypergraph Perfect Matching to Finding Half Layers}
\label{sec: half layer}

In this section we introduce key concepts of half layer and approximate half layer which we will use. We also state the lemmas which allow us to reduce a computation of a perfect matching in certain bipartite hypergraphs to computing these objects.

\begin{definition}%
\label{def: half layer}
    A \emph{half layer} for a bipartite hypergraph $H = (A,B,E)$ with respect to a \emph{state} $(A',B',M)$ and a degree parameter $\Delta$, where the state consisting of a set of activated vertices $A'\subseteq A$, a set of forbidden vertices $B'\subseteq B$, a partial matching $M$, is a subset of edges $Z\subseteq E\setminus M$ that satisfies:
    \begin{enumerate}
        \item $\forall e\in Z$, $e\cap A\in A'$, furthermore, for each vertex $a\in A'$, the number of hyperedges in $Z$ incident to $a$ is at most $\Delta$;
        \item For each $e\in Z$, $e$ is disjoint from $B'$. For any pair of hyperedges $e,e'\in Z$, $B(e)\cap B(e') = \emptyset$;
        \item For each $e'\in M$, there is at most one hyperedge $e$ in $Z$ such that $B(e)\cap B(e') \ne \emptyset$. %
    \end{enumerate}
    The \emph{rank} of the half layer $Z$ is the largest $|e\cap B|$ over $e\in Z$. We say $Z$ is \emph{$r'$-maximal} %
    if it has rank $r'$ and for any $e\in E\setminus Z$ with $|e\cap B|\le r'$, $Z\cup \{e\}$ is not a half layer w.r.t.\ $(A',B',M)$ and $\Delta$.
\end{definition}

The following simple greedy algorithm for constructing an $r'$-maximal half layer should help with understanding the definition.

\begin{proposition}\label{prop:greedy half layer}
Given a bipartite hypergraph $H=(A,B,E)$, a state $(A',B',M)$, and parameters $\Delta$ and $r'$, we can find a half layer $Z$ w.r.t.\ $(A',B',M)$ and $\Delta$ that is $r'$-maximal in time $O(p(H))$. 
\end{proposition}
\begin{proof}
    Initialize $Z \gets \emptyset$. The algorithm maintains, for each vertex in $A$, a counter representing the number of edges in $Z$ incident to it; and the algorithm also maintains a data structure answering for each vertex in $B$,
    (1) whether it belongs to the forbidden vertex set $B'$, and
    (2) whether it is matched in $M$, and if so, by which edge.
    We remark that supporting (2) can be done by a preprocessing algorithm that runs through all edges in $M$ and marks all vertices in the matched edges in the beginning of the algorithm.

    Then, the algorithm iterates over every edge $e \in E$ and checks whether $e \cap A \in A'$, whether $|e\cap B|\le r'$, whether the $A$ vertex in $e$ is incident to fewer than $\Delta$ edges in $Z$, and whether $e$ is disjoint from the current set $B'$.
    If all four conditions are satisfied, the algorithm adds $e$ to the half layer $Z$, and updates the maintained data structures. That is,
    the algorithm increases the counter for the $A$ vertex of $e$ and sets $B' \gets B' \cup B(e)$. Finally, for every vertex $u \in B(e)$, if $u$ belongs to some matching edge $e' \in M$ according to (2), the algorithm updates $B' \gets B' \cup B(e')$ to satisfy the third requirement of being a half layer.
    Notice that each edge $e'\in M$ can only be considered once throughout the algorithm. Thus, 
    the total running time is $O(\sum_{e\in E}|e| + \sum_{e'\in M}|e'|) = O(p)$. 
\end{proof}

The following \Cref{thm: quadratic time hypergraph perfect matching} reduces the hypergraph perfect matching problem into finding maximal half layers. This gives an almost-linear time algorithm for solving the hypergraph perfect matching problem (\Cref{thm: almost linear time hypergraph perfect matching}) for hypergraphs under a strong Haxell condition.

\begin{restatable}{lemma}{quadraticLemma}
\label{thm: quadratic time hypergraph perfect matching}
    Let $H = (A \cup B, E)$ be an $r$-bounded bipartite hypergraph that satisfies the $\varphi$-strong Haxell condition with $\varphi \ge d(n) r^2$ for some parameter $d(n) \ge 4$. Let $T^*$ denote the runtime of an algorithm that computes a maximal half-layer
    for any state with degree parameter $\Delta=d(n)$. Then, there exists an algorithm that computes a perfect matching of $H$ in time $O\left(T^* n^{1/\Omega(\sqrt{\log d(n)})}\right)$.
\end{restatable}

\begin{proof}[Proof of \Cref{thm: almost linear time hypergraph perfect matching}]
We simply plug \Cref{prop:greedy half layer} into \Cref{thm: quadratic time hypergraph perfect matching} with $T^*=O(p)$.
\end{proof}

\paragraph{Half Layer Oracles for Disjoint-Paths Problem.} 
To solve the disjoint-paths problem, an edge oracle has to be designed in order to implicitly solve the hypergraph perfect matching problem.
As mentioned in the discussion below \Cref{def:hyper matching}, each hyperedge corresponds to a path that connects some demand pair.

Consider a maximal half layer $Z$, the hyperedges in $Z$ all together corresponds to a collection of edge-disjoint paths such that, after removing all the edges from the collected paths, there is no short path connecting any unfulfilled demand pairs.
Observe that, a half layer oracle which returns a maximal half layer
as in \Cref{prop:greedy half layer} 
can be na\"ively implemented by repeatedly running BFS and removing all edges from the found path.
This gives a maximal half layer oracle in quadratic time. By plugging $T^*=\tilde{O}(mn\min\{k, \phi^{-1}\})$ and $d(n)=\Theta(\log n)$ into \Cref{thm: quadratic time hypergraph perfect matching}, we obtain an almost-quadratic time algorithm for the disjoint-paths problem (\Cref{thm: routing conductance} Part I) on $\phi$-expanders with a polylogarithmic minimum degree requirement ($\phi^3\delta \ge (35 \log n)^3 k$). As we focus on introducing the half layer oracles in this section, we defer the full proof
to \Cref{sec: almost quadratic time disjoint paths}.

On $\phi$-expanders with a higher minimum degree requirement ($\phi^3 \delta \ge n^{o(1)} k$), an almost-linear time algorithm for disjoint-paths problem can be obtained, with 
a dedicated implementation of half layer oracle that does not always return maximal half layers.
In particular, we introduce the \emph{approximate half layer} in \Cref{def:apx half layer} below.

\begin{definition}\label{def:apx half layer}
    Let $H = (V,E)$ be an $r$-bounded hypergraph. For any $r' \le r$ and $\alpha \ge 1$, a half layer $Z$ with respect to $(A',B',M)$ and $\Delta$ is an \emph{$(r',\alpha)$-approximate half layer} if for any half layer $Z'$ of rank $r'$ with respect to $(A',B',M)$ and $\Delta$, we have $|Z'|\le \alpha|Z|$. %
\end{definition}

In \Cref{sec: almost linear}, we show how to find an approximate half layer above using the multi-commodity flow algorithm of \cite{haeupler2024low}, which incurs approximations in both length and congestion. These correspond, respectively, to the approximations in \Cref{def:apx half layer} for the hyperedge rank and the half-layer size.

Once we have an efficient subroutine to compute an approximate half layer, we obtain an almost linear time hypergraph perfect matching algorithm, for hypergraphs with strong Haxell condition, where the required strength depends on the approximation ratio $\alpha$ we can guarantee when computing our approximate half layer.
We summarize the reduction below in \Cref{lem: fast perfect matching}.

\begin{restatable}{lemma}{linearLemma}
\label{lem: fast perfect matching}
    Let $H=(A, B, E)$ be an $r$-bounded bipartite hypergraph such that the subgraph $H_{r'} = (A, B, E_{r'})$, consisting of edges of $H$ of rank at most $r'$, satisfies the $\varphi$-strong Haxell condition with $\varphi =  24\alpha \cdot d(n) r^2 $ for some parameter $d(n)\ge 4\alpha\ge 4$. Let $\hat T$ denote the runtime of an algorithm that computes an $(r',\alpha)$-approximate half-layer for any state with degree parameter $\Delta=d(n)$. Then, there exists an algorithm that computes a perfect matching in $H$ in time $O\left(\hat{T}n^{1/\Omega(\sqrt{\log d(n)})}\right)$.
\end{restatable}

The proofs of \Cref{thm: quadratic time hypergraph perfect matching,lem: fast perfect matching} are deferred to \Cref{sec:proof hyper matching}.

\section{Finding Disjoint Paths in Expanders Deterministically}
\label{sec: disjoint paths}

In this section, we solve the disjoint-paths problem using the formally defined half layer oracles from the previous section, namely \Cref{thm: quadratic time hypergraph perfect matching,lem: fast perfect matching}.
As mentioned in the introduction (below \Cref{def:hyper matching}), the problem of obtaining short edge-disjoint paths in a graph $G$ reduces to finding a perfect matching in a certain auxiliary $r$-bounded bipartite hypergraph $H$.
The following definition describes these hypergraphs.

\begin{definition}[Demand-Path Hypergraph]\label{def: pair-connectivity}
    Given a graph $G=(V, E_G)$, a set of demand pairs $D$, and a maximum allowed length $r$, the \emph{demand-path hypergraph} $H$ is an $r$-bounded bipartite hypergraph $H = (A,B,E_H)$ defined as follows.
    \begin{enumerate}
        \item The demand pairs in $D$ are in one-to-one correspondence with the vertices in $A$. For each demand pair $(s,t)$, we denote the corresponding vertex in $A$ by $a_{s,t}$.
        \item The edges $E_G$ are in a one-to-one correspondence with the vertices in $B$. For each edge $e\in E_G$, we denote the corresponding vertex in $B$ by $b_e$.
        \item The edge set $E_H$ is built as follows. For each pair $(s,t)\in D$, and each path $P$ in $G$ connecting $s$ and $t$ of length at most $r$, we add a hyperedge $\{a_{s,t}\}\cup \{b_e\mid e\in P\}$ to $E_H$.
    \end{enumerate}
\end{definition}

In \Cref{sec: proving-strong-haxell-condition}, we show that demand-path hypergraphs of expanders satisfy strong Haxell's condition. 
In \Cref{sec: almost quadratic time disjoint paths}, we implement the maximal half layer oracle, obtaining an almost-quadratic time algorithm for the disjoint-paths problem via \Cref{thm: quadratic time hypergraph perfect matching}.
In \Cref{sec: almost linear}, we implement the approximate half layer oracle, obtaining an almost-linear time algorithm via \Cref{lem: fast perfect matching}.

\subsection{Demand-Path Hypergraphs of Expanders Satisfy Strong Haxell}\label{sec: proving-strong-haxell-condition}

The goal of this subsection is to prove demand-path hypergraphs of expanders satisfy strong Haxell's condition, formalized below.

\begin{lemma}\label{lem: pair connectivity has haxell}
    Let $G$ be an $n$-vertex graph with conductance $\phi>0$ and minimum degree $\delta>0$. Let $D$ be a set of demand pairs, with any vertex belonging to at most $k$ pairs. Then, the demand-path hypergraph $H$ of $G,D$ and maximum  length $\lfloor 18 \log (n)/\phi\rfloor$ has the $\frac{\phi \delta}{32k}$-strong Haxell condition. %
\end{lemma}

It is easier to prove \Cref{lem: pair connectivity has haxell} while working with vertex-disjoint demand pairs. We start with the following simple fact.

\begin{fact}
\label{fact: prune set}
    For a set of demand pairs $D = \{(s_i,t_i)\}$ such that each vertex appears in at most $k$ pairs, there exists a subset $C\subseteq D$ of size at least $|D|/2k$ consisting of vertex disjoint pairs. 
\end{fact}
\begin{proof}
    Initially, we set $C = \emptyset$. For each pair in $D$, if it does not intersect any pair already in $C$, we add it to $C$. As each pair in $D$ intersects at most $2(k-1)$ other pairs in $D$, we have $|C|\ge |D|/2k$.
\end{proof}

The helper lemma below shows that, even after deleting a small set of edges, some demand pairs in an expander are still close to each other.

\begin{lemma}
\label{lem: non concurrent disjoint paths}%
Let $G$ be an $n$-vertex graph with conductance $\phi>0$ and minimum degree $\delta>0$.
Fix a set of vertex-disjoint demand pairs $C$. For any set of edges $F$ of size at most $\phi \delta |C|/16$, 
there exists a path of length at most $18\log(n) /\phi$ 
in $G\setminus F$ connecting some pair in $C$.
\end{lemma}

The proof of the above lemma follows immediately from the known fact about expander pruning:

\begin{lemma}[Theorem 1.3 of \cite{SW19}]\label{lem: expander pruning}
Let $G$ be an $n$-vertex graph with conductance $\phi$. For any edge set $F$, let $G' = G\setminus F$. There exists a vertex set $P$ such that  $\sum_{u\in P}\deg(u) \le 8|F|/\phi$ and $G'[V \setminus P]$ has conductance at least $\phi/6$. 
\end{lemma}

\begin{proof}[Proof of \Cref{lem: non concurrent disjoint paths}]
Let $G' = G\setminus F$ and $P$ be the set from \Cref{lem: expander pruning}. 
We have $\delta|P|\le 8|F|/\phi < \delta |C|/2$ because $|F| < \phi \delta |C|/16$. Since $|P|< |C|/2$, there is a demand pair $(s,t)$ from $C$ where $s,t \in G'[V \setminus P]$. Since $G'[V \setminus P]$ has conductance $\phi/6$, its diameter is at most $18\log(n)/\phi$. So the distance between $s$ and $t$ in $G\setminus F$ is also at most $18\log(n)/\phi$.
\end{proof}

Given both helpers, \Cref{fact: prune set} and \Cref{lem: non concurrent disjoint paths}, we are ready to conclude \Cref{lem: pair connectivity has haxell}.

\begin{proof}[Proof of \Cref{lem: pair connectivity has haxell}]
     Let $\varphi = \frac{\phi \delta}{32k}.$ 
     Let $D'\subseteq D$ be a subset of our demand pairs.   
     We now show that any hitting set\footnote{A hitting set for a subset $E_S$ of edges in a bipartite hypergraph $(A,B,E)$ is a set of vertices in $B$ which intersects every edge in $E_S$. In particular, $\tau(E_S)$, equals the minimum size of such a hitting set.}$T\subseteq B$\thatchaphol{we never defined a hitting set yet. (Do it here or in the intro) It might be good to use the notation $\tau$ in this proof once. So that readers can connect it to the definition.}
     \matija{did something along these lines} in $H$ for all edges incident to the set $S = \{ a_{s,t}\mid (s,t)\in D'\}$ has size at least $\varphi|D'|$\shangen{added something in this sentence to clarify a bit}, implying $\tau(E_S) \ge \varphi|S|$. Suppose for contradiction, that $|T|< \varphi |D'|$. 
     Let $F$ be the set of all edges in $G$ corresponding to the elements of $T$. Let $C\subseteq D'$ be a subset of size at least $|D'|/2k$ consisting of vertex disjoint pairs, which exists by \Cref{fact: prune set}. Now, we have $|F|=|T|< \varphi|D'| < 2k\varphi  \cdot |C| = \phi \delta |C|/16$. By \Cref{lem: non concurrent disjoint paths}, there exists a path in $G\setminus F$ of length at most $18\log (n)/\phi$ connecting some pair in $C$. Since this path is disjoint from $F$, it represents an edge in $E_S$ not hit by $T$, a contradiction. %
\end{proof}

\subsection{Disjoint Paths via Maximal Half Layers}

\label{sec: almost quadratic time disjoint paths}

In this subsection, we prove the first case of \Cref{thm: routing conductance}, which establishes an almost-quadratic time algorithm under the assumption that the cube of the conductance times the minimum degree is $\poly\log n$. The following is a restatement of \Cref{thm: routing conductance} Part I.

\begin{theorem}
\label{thm: part I routing conductance}
    Let $G$ be an $n$-vertex $m$-edge graph with  conductance $\phi$ and minimum degree $\delta$. Let $k\ge 1$ be an integer such that $\phi^3\delta\ge (35\log n)^3 k$. Given a set of demand pairs $\{(s_i,t_i)\}$ such that each vertex appears in at most $k$ pairs, there exists a deterministic algorithm which computes in time $mn^{1+o(1)}\min\{k, \phi^{-1}\}$ a set of edge-disjoint paths from $s_i$ to $t_i,$ for every $i$, each of length at most $r=18\log (n)/\phi$.
\end{theorem}

We will prove this theorem by applying \Cref{thm: quadratic time hypergraph perfect matching}, which involves finding a maximal half layer in the demand-path hypergraph.

\paragraph{Finding a Maximal Half Layer.}
Consider the demand-path hypergraph $H$ with respect to $G$, our given set of demand pairs, and the maximum allowed length $r := \lfloor 18\log (n)/\phi\rfloor$.
We first translate the problem of finding a maximal half layer in $H$ with degree parameter $\Delta := 4 \log n$
to the corresponding problem on the original graph $G$. 
Given parameters $(A',B',M, \Delta)$ from \Cref{def: half layer}, let $F$ be the union of edges in all paths that correspond to the hyperedges in the partial matching $M$ together with all edges corresponding to elements in $B'$. 
Observe that finding a maximal half layer is translated to the following:

\begin{problem}\label{prob: maximal in G}
    Finding a maximal set of disjoint paths connecting demand pairs from $A'$ in the graph $G\setminus F$, with the constraint that each pair is connected by at most $\Delta$ paths. 
\end{problem}

We show two different solutions for \Cref{prob: maximal in G}. The first approach is to be greedy using BFS:

\begin{lemma}\label{lem:implement half layer BFS}
    \Cref{prob: maximal in G} can be solved in $O(mnk\Delta)$ time.%
\end{lemma}
\begin{proof}
    Iterate through each demand pair $(s_i, t_i)$ and repeatedly run a BFS from $s_i$ for at most $\Delta$ times.
    As long as there exists a path reaching $t_i$ of length at most $r$, remove the path and add the corresponding hyperedge to $Z'$.
    Since there are at most $nk$ demand pairs, and we repeat at most $\Delta$ times per demand pair, the total running time is $O(mnk\Delta)$.
\end{proof}
The second approach is also simple. Since its description is similar to Dinitz's blocking flow algorithm~\cite{Dinitz06}, we defer the proof to \Cref{sec: proof dinitz}.
\begin{lemma}\label{lem:implement half layer Dinitz}
    \Cref{prob: maximal in G} can also be solved in $O(mnr)$ time.
\end{lemma}

\paragraph{Obtaining Disjoint Paths.}
Given an efficient algorithm for maximal half layer,  we conclude the proof of \Cref{thm: part I routing conductance}.

\begin{proof}[Proof of \Cref{thm: part I routing conductance}]
    By \Cref{lem: pair connectivity has haxell}, $H$ satisfies the $\varphi$-strong Haxell condition with with $\varphi = \frac{\phi \delta}{32k} \ge (4\log n)r^2$.
    Next, \Cref{lem:implement half layer BFS,lem:implement half layer Dinitz} show that, given any state and degree parameter $\Delta = 4\log n$, we can find a maximal half layer in our demand-path hypergraph $H$ with maximum allowed length $r$ in time $O(\min\{mnk\Delta, mnr \})$.  
    Since $\varphi \ge \Delta r^2$, we can apply  \Cref{thm: quadratic time hypergraph perfect matching} and obtain a perfect matching in $H$, corresponding to edge-disjoint paths of length $r$ that connect all demand pairs in $O(\min\{mnk\Delta, mnr \} n^{1/\Omega(\sqrt{\log \log  n})})=O(mn^{1+o(1)}\min\{k,\phi^{-1}\})$ time as desired.
    \qedhere

\end{proof}

\subsection{Disjoint Paths via Approximate Half Layers: Almost-Linear Time}
\label{sec: almost linear}

In this subsection, we aim to prove the second case of \Cref{thm: routing conductance}, which establishes an almost-linear time algorithm under the assumption that the cube of the conductance times the minimum degree is $n^{o(1)}$.

\begin{theorem}
\label{thm: part II routing conductance}
 Let $G$ be an $n$-vertex $m$-edge graph with  conductance $\phi$ and minimum degree $\delta$. Let $k \ge 1$ be an integer for which $\phi^3 \delta > n^{o(1)} k$. Given a set of demand pairs $\{(s_i,t_i)\}$ such that each vertex appears in at most $k$ pairs, there exists a deterministic algorithm, with runtime $m^{1+o(1)},$ computing a set of edge-disjoint paths from $s_i$ to $t_i$ one for every $i$, each of length at most $n^{o(1)}/\phi$.
\end{theorem}

We will prove \Cref{thm: part II routing conductance} by applying \Cref{lem: fast perfect matching}, which involves finding approximate half layers on demand-path hypergraphs.
As mentioned in \Cref{sec:technical overview}, we use the multi-commodity flow algorithm so we start with some preliminaries below.
We note that in our applications we only work with graphs with unit edge lengths and capacities, so we will only define the multicommodity flow in this special case (see e.g.\ \cite{haeupler2024low} for a fully general definition).

\paragraph{Preliminaries on Multicommodity Flow.}
Let $G=(V,E)$ be a graph.
A \emph{(multicommodity) flow} $F$ in $G$ is a function that assigns each simple path $P$ in $G$ a flow value $F(P)\ge 0$. The \emph{value} of the flow is $\val(F)=\sum_{P} F(P)$. The \emph{support} of $F$, denoted by $\supp(F) := \{P: F(P) > 0\}$, is the set of flow-paths. The \emph{congestion} of $F$ on edge $e$ is $\conge_F(e) = F(e)$, where $F(e) = \sum_{P: e\in P} F(P)$ denotes the total flow value of all paths going through $e$. The \emph{congestion} of $F$ is $\conge_F = \max_{e\in E(G)} \conge_F(e)$. The \emph{length} of $F$, denoted by $\leng_F := \max_{P\in \supp(F)} |P|$, is the maximum length of the flow-paths. A \emph{demand} $D:V\times V\to \R_{\ge 0}$ assigns a value $D(u,v)\ge 0$ to each ordered pair of vertices. We say a flow $F$ \emph{routes/satisfies} a demand $D$ if for each $u,v\in V$, $D(u,v) = \sum_{P\textrm{ is a $(u,v)$-path}} F(P)$. The \emph{support} of $D$ is $\supp(D)=\{(u,v)\mid D(u,v)>0\}$. We say $D'$ is a \emph{subdemand} of $D$ if $D'$ is a demand and $D'(u,v)\le D(u,v)$ for all $u,v\in V$.
\\

We apply the low-step multicommodity flow algorithm from \cite{haeupler2024low}, which has an approximation slack on both the length and congestion. Since this oracle will be applied to a subgraph of $G$, we do not assume the underlying graph to be an expander. It is worth noting that even when $G$ is an expander, the state-of-the-art routing (flow) algorithms still incur an $n^{o(1)}$ approximation in either length or congestion \cite{chang2024deterministic, ChangS20, GhaffariL18, GhaffariKS17}\thatchaphol{cite earlier expander routing papers too. Mention this: ``Improving this $n^{o(1)}$ factor in the expander routing problem to $\poly\log(n)$ is the central barrier in the area of fast graph algorithms with wide-ranging applications''. But this discussion explaining why improving $n^{o(1)}$ seems hard should be moved to the intro actually. Here, we should just explain the intuition of theorem below.}\shangen{added!}, if an almost-linear runtime is required.
Improving this $n^{o(1)}$ factor in the expander routing problem to $\poly\log(n)$ is the central barrier in the area of fast graph algorithms with wide-ranging applications.
The following lemma is an instance of Theorem 8.1 from \cite{haeupler2024low} which we will use.

\begin{lemma}%
\label{lem: low step flow}
Let $G=(V,E)$ be an $n$-vertex graph, $D$ an integral demand, $h$ a maximum length bound, and $\epsilon\in ((\log n)^{-c},1)$ for some sufficiently small constant $c$, a tradeoff parameter. Then, there exists an algorithm which returns an integral multicommodity flow $F$ routing a subdemand $D'$, such that
\begin{enumerate}
    \item $F$ has maximum length $\beta\cdot h$ and congestion $\kappa$ for $\beta = \exp(\poly(1/\epsilon))$ and $\kappa = n^{\poly(\epsilon)}$. The support size of $F$ is $(|E|+\supp(D))n^{\poly(\epsilon)}$.
    \item Let $F^*$ be the maximum-value multicommodity flow partially routing $D$ of maximum length $h$ and congestion 1. Then, $\val(F)\ge \val(F^*)$.
\end{enumerate}
The algorithm runs in time $(|E|+\supp(D))\cdot \poly(h)n^{\poly(\epsilon)}$.
\end{lemma}

\paragraph{Finding an Approximate Half Layer.}
Once we obtain the (fractional) flow, we greedily round the solution to an integral one, which corresponds to a set of paths. %
This process is formalized in the following lemma, where the parameters are chosen to balance the approximation ratios for both path length and congestion.

\begin{lemma}
\label{lem: approx disjoint paths}
    Let $G$ be an $n$-vertex graph with $m$ edges. Given a set of demand pairs $C$, suppose there exists a set of $q$ disjoint paths each connecting some pair $(s_i,t_i)\in C,$ and having length at most $h = n^{o(1)}$. There exists an algorithm that runs in time $(m+|C|)n^{o(1)}$ that computes a set of disjoint paths of size at least $q/\alpha$, each connecting a pair $(s_i,t_i)\in C$ of length at most $\beta h$ for $\alpha,\beta = n^{o(1)}$.
\end{lemma}
\begin{proof}
    We shall apply \Cref{lem: low step flow} and round the fractional flow to disjoint paths by greedily picking the paths. Let $\epsilon = 1/\log\log n$, length slack $\beta = 2^{\poly(1/\epsilon)} = n^{o(1)}$, and congestion slack $\kappa = n^{\poly(\epsilon)} = n^{o(1)}$. By assumption, there exists a multicommodity flow of value at least $q$. Applying \Cref{lem: low step flow} with the parameter $\epsilon$, the algorithm computes a flow of value at least $q$ with congestion $\kappa$ and maximum length $\beta\cdot h$. Next, we enumerate every path in the support of $F$, and greedily pick the path if it is disjoint from every path we picked before. We now claim that this computes a set of disjoint paths of size at least $q/(\kappa\beta\cdot h)$, indeed: each path we picked blocks at most $\kappa\beta\cdot h$ amount of flow, as the congestion is $\kappa$ and the maximum length is $\beta\cdot h$. We have the parameters $\beta = n^{o(1)}$ and $\alpha = \kappa\beta\cdot h = n^{o(1)}$. Each of the paths connects a pair $(s_i,t_i)\in C$, and has length at most $\beta\cdot h$. By \Cref{lem: low step flow}, the algorithm runs in $(m+|C|)n^{o(1)}$ time.
\end{proof}

\paragraph{Obtaining Disjoint Paths.}
We now combine all the techniques developed so far to complete the proof of our theorem.  The following is a restatement of \Cref{thm: routing conductance} Part II.

\begin{proof}[Proof of \Cref{thm: part II routing conductance}]
    Take $h = 18\log (n)/\phi$, $\alpha$ and $\beta$ as in \Cref{lem: approx disjoint paths}, $r = \beta\cdot h$, and $d(n)=\alpha =n^{o(1)}$. We shall prove the following two properties so that we can apply \Cref{lem: fast perfect matching}:
    \begin{enumerate}[(1)]
        \item the $\varphi$-strong Haxell condition holds with $\varphi\ge 24 \alpha d(n) r^2$ for the demand-path hypergraph $H$ with maximum length $h$ that corresponds to $G$;
        \item there exists a subroutine that computes $(h,\alpha)$-approximate half layer.
    \end{enumerate}
    
    For (1), by \Cref{lem: pair connectivity has haxell}, $H$ satisfies $\varphi$-strong Haxell condition, with $\varphi = \frac{\phi \delta}{32k} \ge 24 \alpha d(n) r^2,$ as claimed.
    For (2), by \Cref{lem: approx disjoint paths}, there exists an algorithm that computes $(h,\alpha)$-approximate half layer in the $r$-bounded demand-path hypergraph $H$, which runs in $(m+nk)n^{o(1)} \le (m+n\delta\phi^3)n^{o(1)}\le m^{1+o(1)}$ time. 
    
    With these conditions established \Cref{lem: fast perfect matching} shows that there exists an algorithm with running time $m^{1+o(1)}$ that computes a hypergraph perfect matching in $H$, which corresponds to the set of disjoint paths, as desired.
\end{proof}

\section{Hypergraph Matching Algorithms: Proofs of \texorpdfstring{\Cref{thm: quadratic time hypergraph perfect matching,lem: fast perfect matching}}{}}
\label{sec:proof hyper matching}

In this section, we aim to prove our key technical tools: \Cref{thm: quadratic time hypergraph perfect matching,lem: fast perfect matching}. We begin with some preliminaries in \Cref{sec: preliminaries of matching}.
In \Cref{sec: the augmenting algorithm}, we describe the perfect matching algorithm that implements the half-layer oracle with \emph{maximal half layers}.
These first two subsections essentially follow the framework of \cite{Annamalai18} with the slightly modified \emph{alternating forest} (see \Cref{def: alternating-forest}) so that we can plug in the half layer oracles.
In \Cref{sec: perfect-matching-via-building-approximate-layers}, we introduce a new perfect matching algorithm that implements the half layer oracle with \emph{approximate half layers}.
We include some analysis in the two previous sections so that, in \Cref{sec: iterations-and-depth-of-alternating-forests-of-both-algorithms}, we can bound the number of iterations and the maximum depth of alternating forests needed for both versions of the algorithms.
These two quantities determine the efficiency of the algorithms. Given these bounds, we finally conclude our main lemmas in \Cref{sec: fast matching algos}.

\subsection{Section Preliminaries}
\label{sec: preliminaries of matching}

Consider a matching $M$ in a normal bipartite graph $G=(A,B,E)$. For each non-matched edge $e\notin M$, there is a unique matched edge $f\in M$ that ``blocks'' $e$ from the $B$ side, i.e., $f\cap e \cap B \neq \emptyset$. In hypergraphs, the blocking edges are not unique anymore. This motivates the following definition.
\begin{definition}[Blocking Edges]
    Given a partial matching $M$ in a bipartite hypergraph $H=(A,B,E)$ we define the set of edges \emph{blocking} a given edge $e\in E$ as
    \[ \{f\in M\mid f\cap e\cap B \neq \emptyset\}.
    \]
    The set of blocking edges of $X \subseteq E$ is defined as the union of the blocking edges of each $e \in X$.
\end{definition}
In other words, the set of blocking edges consists of edges already in $M$ that prevent us from adding $e$ to it because of an intersection within the $B$ part. We note that we do not take into account possible intersections in the $A$-part here.
In the following we define what a layer means in the alternating forest. A layer consists of a half layer that we defined in \Cref{def: half layer}, combined with the set of blocking edges of the half layer.

\begin{definition}[Layer]
    A \emph{layer} $L$ for a bipartite hypergraph $H = (A,B,E)$ w.r.t.\ a state $(A',B',M)$ and parameter $\Delta$ is a tuple $(X,Y)$ where $X$ is a half layer w.r.t.\ $(A',B',M)$ and parameter $\Delta$, and $Y\subseteq M$ is precisely the set of blocking edges of $X$. We say $L$ is \emph{maximal} if $X$ is a maximal half layer. We say $L$ is \emph{$(r',\alpha)$-approximate layer} if $X$ is an $(r',\alpha)$-approximate half layer.
\end{definition}

One should think of a layer as defined by its half layer $X,$ which is itself a collection of edges which (1) are disjoint within the $B$ part, (2) are currently unused by $M$, and (3) have their sets of blocking edges mutually disjoint. %
The union of these blocking sets is then taken to be the set $Y$ above. The motivation behind this definition is that if one manages to free all the edges in $Y$, then all the edges in $X$ become eligible to be added to the matching, at least from the perspective of the $B$-part. If their vertex on the $A$ side is already used, we are instead at least able to do a switch.

For convenience, for any sequence of sets $S_1,\ldots, S_\ell$, we define  $S_{\le \ell}:=S_1 \cup S_2\cup \ldots \cup S_\ell$.

\begin{definition}[Alternating Forest]\label{def: alternating-forest}
    An \emph{alternating forest} $T$ for a bipartite hypergraph $H = (A,B,E)$ with respect to a partial matching $M$ and a degree parameter $\Delta$ %
    is a tuple 
    $(L_0,\ldots,L_\ell)$ such that:
    \begin{itemize}
        \item $L_0 := (\emptyset, A_0)$ 
        where $A_0 = A\setminus (\bigcup_{e\in M}e)$ contains unmatched $A$-vertices;
        \item For all $1\le i\le \ell$, $L_i=(X_i,Y_i)$ is a layer with respect to $(A(Y_{i-1}),B(X_{\le i-1}\cup Y_{\le i-1}), M, \Delta)$.
    \end{itemize}
We denote the prefix alternating forest of $(L_0, \ldots, L_\ell)$ by $\{L_{\le t}\}:=(L_0, \ldots, L_t)$ for any $0\le t \le \ell$. We note that throughout the paper $X_i$ and $Y_i$ will always denote the first and second coordinate of $L_i,$ and we often do not specify this in an attempt to keep the notation under control. 
\end{definition}

For example, suppose at some point we find an edge in $X_\ell$ which does not have any edges in the matching intersecting it within the $B$ part. We can then attempt to add these edges to the matching, but to do so we need to remove the $Y_{\ell-1}$ edges which intersect them in the $A$ vertices. Just doing so has not increased the size of our matching unless $\ell=1$, but it removed some edges previously in $M$ from $Y_{\ell-1}$ which might result in additional edges from $X_{\ell-1}$ now not being ``blocked'' within the $B$ part, which allows us to repeat. In reality, we will need to again grow the tree occasionally, but can ensure we are making progress towards actually being able to add some edges in $X_1$ which will match some new vertices from $Y_0$.

\subsection{Perfect Matching via Building Maximal Layers}
\label{sec: the augmenting algorithm}

Before introducing the main algorithm, we describe the subroutine for building the next layer in the alternating forest. The input is an alternating forest $T$, and a pair of sets $X',Y'\subseteq E$ that serve as the initial values for the next layer. Usually, $X'$ and $Y'$ will be empty, and the subroutine simply computes the next layer from scratch. When modifying an already built forest, however, it will be convenient to be able to reuse the part of the layer already built rather than rebuilding it from scratch. %

We first present a naive approach to a build layer subroutine. Here, we simply add edges to $X'$ greedily so long as they maintain the variety of properties we require of the next layer of our forest, including the maximum degree (controlled by the parameter $\Delta$) condition, which we need in order to maintain control of the process in the analysis stage. The notion of an \emph{addable edge} captures these requirements on an edge. 

\begin{definition}[Addable Edge]
    Let $T = (L_0,\ldots,L_\ell)$ be an alternating forest w.r.t.\ a partial matching $M$, 
    where $L_i=(X_i,Y_i),$ and $X',Y'\subseteq E$. We say $a\in A$ has a \emph{$\Delta$-addable edge} $e$ w.r.t.\ %
    $(T,X',Y')$ if
    \begin{itemize}
        \item $a$ is contained in fewer than $\Delta$ edges in $X'$, and
        \item $e$ is disjoint from $B(X_{\le\ell}\cup Y_{\le\ell}\cup X'\cup Y')$.
    \end{itemize}
    In addition, if no edges from $M$ block $e$ then we call $e$ \emph{immediately addable} w.r.t.\ $(T,X',Y')$.
\end{definition}

\begin{algorithm}[H]
\caption{Naive Approach for Building a Layer.}
\label{alg: build layer}
\begin{algorithmic}
\Procedure{BuildLayer}{$T=(L_0, \ldots, L_\ell),X',Y',\Delta$}
\While{$\exists$ a $\Delta$-addable edge $e$ w.r.t.\ $(T,X',Y')$ for some $a\in A(Y_{\ell})$} \label{line: addable edge}
\State $X'\gets X'\cup \{e\}$.
\State $Y'\gets Y'\cup \{f\in M\mid f\cap e\cap B \neq \emptyset \}$.
\EndWhile
\State \Return $(X',Y')$.
\EndProcedure
\end{algorithmic}
\end{algorithm}

We first show that the \Call{BuildLayer}{} subroutine extends a layer to a maximal layer.

\begin{proposition}
\label{prop: build layer}
    Given a bipartite hypergraph $H=(A,B,E),$ an alternating forest $T=(L_0,\ldots,L_{\ell})$ for $H$ w.r.t.\ a partial matching $M$, sets of edges $X',Y'\subseteq E$ such that $(X',Y')$ is a layer w.r.t.\ $(A(Y_{\ell}),B(X_{\le\ell}\cup Y_{\le\ell}),M),$ and a parameter $\Delta$, the \textsc{BuildLayer}{$(T,X',Y',\Delta)$} procedure returns a maximal layer w.r.t.\ $(A(Y_{\ell}),B(X_{\le\ell}\cup Y_{\le\ell}),M)$ and parameter $\Delta$.
\end{proposition}

\begin{proof}
    Let $(X' \cup Z, Y'')$ be the output of \Call{BuildLayer}{$T,X',Y',\Delta$}. 
    $Z$ is a half layer w.r.t.\ state $(A(Y_{\ell}),B(X_{\le\ell}\cup Y_{\le\ell}),M)$ and parameter $\Delta$ by the definition of addability, and it is maximal since the procedure repeatedly attempts to add additional edges so long as that is feasible. 
    Note also that throughout the procedure the current $Y'$ contains all blocking edges of the current $X'$. Hence, $Y''$ contains all blocking edges of $X' \cup Z$.
\end{proof}

We now present the main algorithm for the hypergraph perfect matching. %

\begin{algorithm}[H]
\caption{Hypergraph Perfect Matching.}
\label{alg: hypergraph perfect matching}
\begin{algorithmic}[1]
\Procedure{HyperGraphMatching}{$H=(A, B, E)$}
\State Set parameters $\Delta$ and $\mu$. 
\State Initialize the partial matching $M = \emptyset$. 
\State Initialize $L_0$ in the alternating forest $T$ by setting $(X_0,Y_0)\gets (\emptyset,A)$.
\State Set $\ell = 0$. %
\While{$Y_0$ is not empty} \label{line: main loop begin} \Comment{Main loop.}
\State $(X_{\ell+1},Y_{\ell+1})\gets \Call{BuildLayer}{T,\emptyset,\emptyset, \Delta}$. \label{line: build layer main}

\State Add the new layer $L_{\ell+1}:=(X_{\ell+1},Y_{\ell+1})$ to $T$.
\State Increment $\ell$ to $\ell + 1$.
\State $(T,M,\ell)\gets\Call{CollapseForest}{T,M,\ell}$. \Comment{See \Cref{alg: collapse forest}.}
\EndWhile \label{line: main loop end}
\State \Return $M$.
\EndProcedure
\end{algorithmic}
\end{algorithm}

We refer to the \textbf{while} loop from line~\ref{line: main loop begin}-\ref{line: main loop end} as the main loop of the algorithm. We note that $Y_0$ will always contain all currently unmatched vertices in $A$ and that $\ell$ will always point to the current forest's last layer. 
In our analysis, we aim to bound the number of the iterations of the main loop under the strong Haxell condition. In an iteration, the \textsc{BuildLayer} subroutine will be the bottleneck, and we shall carefully analyze the runtime afterwards.
One can think of the above algorithm as repeatedly adding new maximal layers to the current forest until certain favorable conditions arise (detected by the \textsc{CollapseForest} procedure) in which case \textsc{CollapseForest} performs certain ``switches'' based on the current alternating tree which will either allow us to reach a ``better'' state or even ideally allow us to increase the size of the matching. We now describe precisely what the \textsc{CollapseForest} procedure does in \Cref{alg: collapse forest}.

\begin{algorithm}[H]
\caption{Collapse the Alternating Forest.}
\begin{algorithmic}[1]
\label{alg: collapse forest}
\Procedure{CollapseForest}{$T=(L_0, \ldots, L_\ell),M,\ell$}
\While{$X_\ell$ contains more than $\mu|X_\ell|$ immediately addable edges} \label{line: swap criteria}
\For{each $f\in Y_{\ell-1}$ such that $\exists$ immediately addable edge $e\in X_\ell$ for $a\in A\cap f$} \label{line: collapse begin} \label{line: swap begin}
    \State $M\gets M\setminus \{f\}\cup \{e\}$. \label{line:switch}
    \State $Y_{\ell-1}\gets Y_{\ell-1}\setminus \{f\}$. \label{line:updateY}%
\EndFor \label{line: swap end}
\State Discard layer $L_{\ell}$ from $T$.
\If{$\ell \ge 2$}
\State $(\hat X_{\ell-1},\hat Y_{\ell-1})\gets \Call{BuildLayer}{\{L_{\le \ell-2}\},X_{\ell-1},Y_{\ell-1},\Delta}$. \Comment{Superpose-build.}\label{line: superpose build begin}
\State \textbf{If} $|\hat X_{\ell-1}|\ge (1+\mu)|X_{\ell-1}|$, \textbf{then} $X_{\ell-1}\gets \hat X_{\ell-1}$, $Y_{\ell-1}\gets \hat Y_{\ell-1}$. \label{line: superpose build end}
\EndIf
\State Decrement $\ell$ to $\ell - 1$. \label{line: collapse end}
\EndWhile
\State \Return $(T,M,\ell)$.
\EndProcedure
\end{algorithmic}
\end{algorithm}
\noindent We note that in both lines~\ref{line: swap criteria} and \ref{line: swap begin} the immediately addable edges are w.r.t.\ $(\{L_{\le \ell-1}\},\emptyset,\emptyset).$ 

We refer to lines~\ref{line: swap begin}–\ref{line: swap end} as a \emph{swap} operation, lines~\ref{line: superpose build begin}–\ref{line: superpose build end} as a \emph{superpose-build} operation, and lines~\ref{line: collapse begin}-\ref{line: collapse end} as a \emph{collapse} operation. A swap operation is performed only when it significantly alters the alternating forest by freeing a large number of matching edges in the previous layer. This restriction, which also applies to the superpose-build operation, is imposed to facilitate the runtime analysis.

It is instructive and could help with building intuition to consider what happens if we call the \textsc{CollapseForest} procedure when $\ell=1$. The main difference compared to $\ell \ge 2$ is that $L_0=(X_0,Y_0)$ is not an actual layer. Namely, $X_0$ is always empty, while $Y_0$ contains all unmatched vertices (in contrast to $Y_i$ for any $i \ge 1$ which consists of edges from $M$). Here, if we find that at least $\mu$ proportion of $X_1$ edges are immediately addable we may simply add them to the matching (done in Line~\ref{line:switch}\footnote{Note here that $f\in Y_0$ is a vertex so never belongs to $M$ and is hence not removed from $M$.}) and their $A$ vertices get removed from $Y_0$ (done in Line~\ref{line:updateY}), since they are not matched.

\subsection*{Analysis}
\label{sec: analyze augmenting algo}

The following two mostly immediate propositions will be useful in the analysis, they essentially match Propositions 4.1 and 4.2 from \cite{Annamalai18}. 

\begin{proposition}[Proposition 4.1 in \cite{Annamalai18}]
\label{lem:analysis-same-layer-ratio}
    At the beginning of each main loop in \textsc{HypergraphMatching}, all $L_t=(X_t,Y_t)$ are non-collapsible, i.e.\ $X_t$ contains no more than $\mu|X_t|$ immediately addable edges w.r.t.\ $(\{L_{\le t-1}\},\emptyset,\emptyset)$, for all $0 \le t \le \ell$. Moreover, for each $0\le t \le \ell$, we have $|Y_t| \ge (1-\mu)|X_t|$. %
\end{proposition}

\begin{proof}[Proof.]
The first part of \Cref{lem:analysis-same-layer-ratio} holds since the \textbf{while} loop in \textsc{CollapseForest} runs until $L_\ell$ becomes non-collapsible. Since we never modify layers except the currently largest one this implies that all lower layers are non-collapsible too.
The second part follows from the first since any edge in $X_\ell$ which is not immediately addable must have at least one blocking edge in $Y_\ell$ (which are disjoint between distinct edges). So, since we never proceed to add an additional layer until the first condition is satisfied we know the second one holds for all $1\le t \le \ell$. The case of $t=0$ is also ensured since $X_0$ is always empty. 
\end{proof}

The following proposition roughly speaking states that all our layers $L_t=(X_t,Y_t)$ are not far from being maximal layers with respect to the sub-forest consisting of the layers up to $L_t$.
\begin{proposition}[Proposition 4.2 in \cite{Annamalai18}]
\label{lem: superpose ub}
    At the beginning of each main loop, for any $L_t=(X_t,Y_t)$, $0\le t\le \ell,$ we have
    $|\hat X_t| \le (1+\mu)|X_t|$, here $(\hat X_t, \hat Y_t) := \textsc{BuildLayer}(\{L_{\le t-1}\}, X_t, Y_t, \Delta)$. 
\end{proposition}

\begin{proof}
If $L_t$ was created by a superpose-build operation then $\hat{X}_t=X_t$. Otherwise, the statement is ensured by the failure of the superpose-build criterion in line~\ref{line: superpose build end} of \Cref{alg: collapse forest}.
\end{proof}

Throughout the paper, we fix $\mu = \frac{1}{10}$. The following lemma applies the strong Haxell condition to derive the growth rate of the alternating forest, which will in turn be used to upper-bound its depth. As discussed earlier, each time we perform a reconstruction operation, such as a swap or superpose-build, a certain potential vector of the forest grows (so will in particular always be different for different alternating forests we encounter). The depth bound yields an upper bound on the total number of viable potential vectors which in turn bounds the number of main loop iterations we may encounter. The next two lemmas will be used to establish the growth of the potential vector and help control the number of viable potential vectors. Its proof follows along similar lines as that of \cite[Lemma 4.5]{Annamalai18}. For our parameter regime, we provide a slightly simpler proof.  
\begin{lemma}\label{lem:analysis-expand-by-delta}
    Let $H$ be an $r$-bounded bipartite hypergraph satisfying the $\varphi$-strong Haxell condition for $\varphi = \Delta r^2$ with $\Delta \ge 4$. %
    Let $L_i=(X_{i},Y_{i})$ for some $1\le i\le \ell$ be a layer of $T$ at the start of the main loop of \textsc{HypergraphMatching}$(H)$. Then,
    \[
        |X_{i}| > \frac{\Delta}{10} \cdot |Y_{\le i-1}|.
    \]
\end{lemma}

\begin{proof}
    Let us first prove that when a layer $(X_{\ell+1},Y_{\ell+1})$ is first created as the output of \textsc{BuildLayer}$(T,\emptyset,\emptyset, \Delta)$ it satisfies the desired inequality. 
    Assume for the sake of contradiction, that $|X_{\ell+1}| \le \frac{\Delta}{10}|Y_{\le \ell}|$.
    Let $(\hat X_t, \hat Y_t) := \textsc{BuildLayer}(\{L_{\le t-1}\}, X_t, Y_t, \Delta)$ for all $t \le \ell$.
    Let $\hat S_t:=A(\hat X_t\setminus X_t)$, so the set of vertices in $A(Y_{t-1})$ which would have an addable edge added to $X_t$ if we were to perform a superpose-build operation (if we actually performed it, this would be an empty set) on layer $L_t$ while ignoring the layers $L_{t+1},\ldots,L_{\ell+1}$.
    Let $S'\subseteq A$ be the set of vertices in $A$ which have degree precisely $\Delta$ in $X_{\le \ell+1}$.
    We consider the set $S\subseteq A$ where $S = A(Y_{\le \ell})\setminus (S' \cup \hat S_{\le \ell})$. 
    Since the degree of any $a\in S$ is strictly less than $\Delta$ in $X_{\le \ell+1}$,
    we know that any edge $e$ containing $a$ 
    intersects $B(X_{\le \ell+1}\cup Y_{\le \ell+1})$, or $B(\hat X_{ \le \ell} \cup \hat Y_{\le \ell})$ (where we are also using the maximality of each $(\hat X_t, \hat Y_t)$ guaranteed by \Cref{prop: build layer}). 
    As $X_{\le \ell}\cup Y_{\le \ell}\subseteq \hat X_{\le \ell}\cup \hat Y_{\le \ell}$, we get
    \begin{align*}
        \tau(E_S) &\le |B(X_{\ell+1} \cup Y_{\ell+1})| + |B(\hat X_{\le \ell}\cup \hat Y_{\le\ell})| \\
        &\le |X_{\ell+1}|r^2 + |\hat X_{\le \ell}|r^2.
    \end{align*}
    The second inequality follows from a basic counting argument: each hyperedge in $X_{\ell+1}$ can be blocked by at most $r$ edges in $Y_{\ell+1}$, and each of those blocking edges can touch at most $r - 1$ vertices in $B \setminus B(X_{\ell+1})$. The same argument applies to the second term. 

    Next, we apply \Cref{lem: superpose ub} and \Cref{lem:analysis-same-layer-ratio} to upper bound the second term.
    \[
        |\hat{X}_{\le \ell}| \le \frac{11}{10} |X_{\le \ell}| \le \frac{11}{9} |Y_{\le \ell}|.
    \]
    
    Combining with our initial assumption $|X_{\ell+1}| \le \frac{\Delta}{10}|Y_{\le \ell}|$, we get 
    \begin{align}\label{(*)}
    \begin{split}
        \tau(E_S) & \le \left( \frac{\Delta}{10} + \frac{11}{9} \right) r^2 |Y_{\le\ell}| \le \frac12 \Delta r^2 |Y_{\le\ell}|
    \end{split}
    \end{align}
    since $\Delta\ge 4$. On the other hand, applying again the initial assumption and \Cref{lem:analysis-same-layer-ratio,lem: superpose ub}, %
    \begin{align*}
    |S| &\ge |A(Y_{\le \ell})|  -  |S'| -  |\hat S_{\le \ell}| \\
    &\ge |Y_{\le \ell}| - \frac{1}{\Delta}|X_{\le \ell+1}| - \frac{1}{10}|X_{\le \ell}|\\
    & \ge |Y_{\le \ell}| \left(1 - \frac{1}{\Delta}\cdot\left(\frac{\Delta}{10} + \frac{10}{9}\right) - \frac{1}{10}\cdot \frac{10}{9}\right) > \frac{1}{2}|Y_{\le \ell}|. 
    \end{align*}
    By our $\varphi$-strong Haxell condition, this implies that 
    \begin{align*}
        \tau(E_S) > \frac{1}{2} \Delta r^2 |Y_{\le \ell}|,
    \end{align*}
    which contradicts \Cref{(*)}.

    Finally, note that after a layer $(X_{i},Y_{i})$ is first created, by the above argument it satisfies the desired inequality. This inequality remains true so long as layer $i$ exists, which is so long as $\ell \ge i+1$. Indeed, the size of $X_i$ only changes if $\ell=i+1$ in which case it might increase in a superpose build step (or remain the same), while so long as $\ell \ge i+1$ the sets $Y_{\le i-1}$ do not get modified. So the conclusion the lemma holds. 
\end{proof}

\subsection{Perfect Matching via Building Approximate Layers}
\label{sec: perfect-matching-via-building-approximate-layers}
We now turn to the faster algorithm which only has access to an approximate build-layer procedure, we require a bit more careful analysis of the growth rate of the alternating forest. We first give the formal definition of an approximate build-layer procedure.

\begin{definition}[\textsc{ApproxBuildLayer}]
    Given a bipartite hypergraph $H=(A,B,E),$ an alternating forest $T=(L_0,\ldots,L_{\ell})$ for $H$ w.r.t.\ a partial matching $M$, sets of edges $X',Y'\subseteq E,$ such that $(X',Y')$ is a layer with respect to $(A(Y_{\ell}),B(X_{\le\ell}\cup Y_{\le\ell}),M),$ and a parameter $\Delta$, the $(r',\alpha)$-approximate build-layer procedure \Call{ApproxBuildLayer}{$T,X',Y',\Delta$} returns a layer $(X'\cup Z, Y'\cup W)$ w.r.t.\ $(A(Y_{\ell}),B(X_{\le\ell}\cup Y_{\le\ell}),M)$ such that $Z$ is an $(r',\alpha)$-approximate half layer w.r.t.\ $(A(Y_{\ell}),B(X_{\le\ell}\cup Y_{\le\ell})\cup B(X'\cup Y'),M)$ and parameter $\Delta$.
\end{definition}

We remind the reader that $Z$ being an $(r',\alpha)$-approximate half layer means that it is not much smaller (by at most an $\alpha$ factor) than any half layer for the same state which is restricted to use only edges of rank at most $r'.$\footnote{We stress that, this in particular does not place any restriction on the rank of edges in $Z$ itself.}

The faster algorithm we use here, which we will refer to as \textsc{FasterHypergraphMatching}, simply replaces both uses of \Call{BuildLayer}{} in \Call{HyperGraphMatching}{} (one in the main loop and one in the \Call{CollapseForest}{}) with \Call{ApproxBuildLayer}{}. 
We note that \Cref{lem:analysis-same-layer-ratio} remains true under this modification since the argument behind it does not involve what these procedures do. On the other hand, we will need to tweak \Cref{lem: superpose ub}. With this in mind we introduce the following definition of an \emph{$r'$-maximal extension}. 

\begin{definition}
    Given a bipartite hypergraph $H=(A,B,E),$ sets of vertices $A'\subseteq A, B'\subseteq B$, a partial matching $M$, sets of edges $X',Y'\subseteq E$ such that $(X',Y')$ is a layer w.r.t.\ $(A',B',M),$ we say a layer $(X'\cup Z, Y'\cup W)$ is an \emph{$r'$-maximal extension} of $(X',Y')$ w.r.t.\ $(A',B',M),$ and parameter $\Delta$, if $Z$ is an $r'$-maximal half layer w.r.t.\ $(A',B'\cup B(X'\cup Y'),M)$ and parameter $\Delta$.
\end{definition}

We note that this notion plays a similar role to the output of \textsc{BuildLayer} (which we showed is actually in a certain sense maximal in \Cref{prop: build layer}). We stress however that this notion is only going to be used (as a sort of benchmark) in the analysis and that our algorithm never actually computes it, since this would be computationally too expensive.

We are now ready to state our analogue of \Cref{lem: superpose ub} in the approximate build layer setting.

\begin{proposition}
\label{lem: superpose ub approx}
    At the beginning of each main loop of \textsc{FasterHypergraphMatching}, for all $0\le t\le \ell$ we have $|\hat X_t'| \le (1+\alpha+\alpha\mu)|X_t|$, where $(\hat X_t', \hat Y_t')$ is any $r'$-maximal extension of $(X_t,Y_t)$ w.r.t.\ $(A(Y_{t-1}),B(X_{\le t}\cup Y_{\le t}),M)$ and parameter $\Delta$. 
\end{proposition}

\begin{proof}
    First note that $\hat X_t,$ which is an output of the approximate build-layer subroutine for \linebreak $(\{L_{\le t-1}\} ,X_t,Y_t,\Delta),$ contains an $(r',\alpha)$-approximate half layer w.r.t.\ $(A(Y_{t-1}),B(X_{\le t}\cup Y_{\le t}),M)$ and $\Delta$, so $|\hat X_t'| \le |X_t|+\alpha |\hat X_t|,$ by the definition of the $(r',\alpha)$-approximate half layer. Since regardless of the outcome of the superpose-build criteria at line~\ref{line: superpose build end} of \Call{CollapseForest}{} we have $|\hat X_t| < (1+\mu)|X_t|$ we conclude $|\hat X_t'| \le (1+\alpha+\alpha \mu) |X_t|$. 
\end{proof}

The following lemma is the analogue of \Cref{lem:analysis-expand-by-delta} for our approximate build-layer setting. 

\begin{lemma}\label{lem: approx analysis-expand-by-delta}
    Let $H=(A, B, E)$ be an $r$-bounded bipartite hypergraph, and suppose that $H_{r'} = (A, B, E_{r'})$, consisting of all the edges of $H$ of rank at most $r'$, satisfies the $\varphi$-strong Haxell condition with $\varphi =  24\Delta \alpha  r^2, $ for some parameter $\Delta\ge 4\alpha\ge 4$. %
    Then, if $L_i=(X_{i},Y_{i}),$ for some $1\le i\le \ell$, is a layer of $T$ at the start of the main loop of \textsc{FasterHypergraphMatching}(H), we have $|X_{i}| > \frac{\Delta}{10} |Y_{\le i-1}|$.
\end{lemma}

\begin{proof}
    Similarly as in the proof of \Cref{lem:analysis-expand-by-delta} it is enough to show the desired inequality holds at the time when the layer $(X_{\ell+1},Y_{\ell+1})$ is first created.
    Assume for the sake of the contradiction that $|X_{\ell+1}| \le \frac{\Delta}{10}|Y_{\le \ell}|$. Let $\hat S_t$ be the set of $A$ vertices from $Y_{t-1}$ that would become a full degree vertex (degree $\Delta$) if we were to perform a superpose-build operation in $H$ by computing an $r'$-maximal extension $(\hat X'_{t}, \hat Y'_{t})$ of layer $L_t$ w.r.t.\ $(A(Y_{t-1}),B(X_{\le t-1}\cup Y_{\le t-1}),M)$ and parameter $\Delta$ while ignoring the layers $L_{t+1},\ldots,L_{\ell+1}$.
    Let $S'\subseteq A$ be the set of all vertices in $A$ which have degree $\Delta$ in $X_{\le \ell+1}$.
    We consider the set $S\subseteq A$ where $S = A(Y_{\le \ell})\setminus (S' \cup \hat S_{\le \ell})$.  From the way we construct $S$, we know that for each $a\in S$, there is no edge of rank at most $r'$ containing $a$ that is disjoint from part $B$ vertices of the alternating forest and from the part $B$ vertices introduced in the superpose-build operation. Thereby, the hitting set w.r.t.\ $H_{r'}$ can be upper bounded by
    \begin{align*}
        \tau(E_S) &\le |B(X_{\ell+1} \cup Y_{\ell+1})| + |B(\hat X'_{\le\ell}\cup \hat Y'_{\le\ell})| \\
        &\le |X_{\ell+1}|r^2 + |\hat X'_{\le \ell}|r^2.
    \end{align*}
     By \Cref{lem: superpose ub approx} and \Cref{lem:analysis-same-layer-ratio}, we have
    \[
        |\hat{X}_{\le \ell}'| \le \left(1+\frac{11}{10}\alpha\right) |X_{\le \ell}| \le \left(\frac{10}{9}+\frac{11}{9}\alpha\right) |Y_{\le \ell}|\le \frac{7}{3}\alpha |Y_{\le \ell}|.
    \]  
    Since we assumed towards a contradiction that $|X_{\ell+1}| \le \frac{\Delta}{10}|Y_{\le \ell}|$, we have an upper bound
    \begin{align}\label{(**)}
    \begin{split}
        \tau(E_S) & \le \left( \frac{\Delta}{10\alpha} + \frac{7}{3} \right)\alpha r^2 |Y_{\le\ell}| \le \frac34\Delta\alpha r^2 |Y_{\le\ell}|
    \end{split}
    \end{align}
    as $\Delta\ge 4\alpha$. On the other hand, $|\hat S_{\le \ell}|\le \frac{1}{\Delta}|\hat X'_{\le\ell}|\le \frac{1}{\Delta}\cdot \frac{7}{3}\alpha|Y_{\le \ell}|$, which combined with our initial assumption, \Cref{lem:analysis-same-layer-ratio}, and \Cref{lem: superpose ub approx} gives
    \begin{align*}
    |S| &\ge |A(Y_{\le \ell})|  -  |S'| -  |\hat S_{\le \ell}| \\
    &\ge |Y_{\le \ell}| - \frac{1}{\Delta}|X_{\le \ell+1}| - \frac{1}{\Delta}|\hat X'_{\le\ell}|\\
    & \ge |Y_{\le \ell}| \left(1 - \frac{1}{\Delta}\cdot\left(\frac{\Delta}{10} + \frac{10}{9}\right) - \frac{\alpha}{\Delta}\cdot \frac{7}{3}\right) > \frac{1}{30}|Y_{\le \ell}|
    \end{align*}
    By $\varphi$-strong Haxell condition on $H_{r'}$, this implies that 
    \begin{align*}
        \tau(E_S) > \frac{5}{6} \Delta\alpha r^2 |Y_{\le \ell}|.
    \end{align*}  
    A contradiction to \eqref{(**)}.
\end{proof}

\subsection{Iterations and Depth of Alternating Forests of Both Algorithms}
\label{sec: iterations-and-depth-of-alternating-forests-of-both-algorithms}

The following lemma allows us to bound the number of iterations of the main loop in terms of the maximum depth of the alternating forest throughout the algorithm. We note that here and in the following proposition, we are assuming the hypergraph we run our algorithm on satisfies the respective strong Haxell condition required by \Cref{lem:analysis-expand-by-delta} or \Cref{lem: approx analysis-expand-by-delta}, respectively.

\begin{lemma}
\label{lem: upper bound main loop iterations}
Let $\ell_{\max}$ denote the maximum number of layers in an alternating forest $T$ considered by the algorithm \textsc{HypergraphMatching} or \textsc{FasterHypergraphMatching}. Then, the total number of iterations of the main loop in the algorithm is at most $2^{\sqrt{O(\ell_{\max}^2+\ell_{\max}\log n)}}$.
\end{lemma}

\begin{proof}
Given the current alternating forest $\{L_0\ldots,L_\ell\},$ with $L_t=(X_t,Y_t),$ for $0\le t \le \ell$ we define its signature vector $\Psi(L_0, L_1, \ldots, L_\ell) := (\psi_0, \psi_1, \ldots, \psi_\ell, \infty)$, where 
\[
\psi_t = \left( -\left\lfloor \log_{1.01}(5^{2t} |X_t|)\right\rfloor, \left\lfloor \log_{1.01}(5^{2t+1} |Y_t|)\right\rfloor \right),
\]
where we define $\log 0 = 0$.

\noindent We consider two cases based on whether a collapse operation occurred during this iteration.
\begin{enumerate}
    \item If no collapse operation occurred, then in this iteration the algorithm only performed the build-layer (or approximate build-layer in case of \textsc{FasterHypergraphMatching}) operation, where a new layer is built. Therefore, $\psi_{\ell+1}$ is inserted at the end of the current signature vector and hence the lexicographic value of the signature decreased.
    \item Let $t$ be the lowest (smallest) layer that was collapsed in this iteration. According to either algorithm, there are two subcases.
    \begin{itemize}
        \item If a successful superpose-build was performed at layer $t-1$, then the size of $X_{t-1}$ increased by at least a factor of $1+\mu$. In this case the value of $\psi_{t-1,0}$ decreased by at least one.
        \item Otherwise, the layer $t-1$ was not modified in the superpose-build phase. In this case, $X_{t-1}$ was not changed during this iteration while $Y_{t-1}$ has changed due to the collapse on $L_t$. By the requirement triggering the collapse, there must be $\mu|X_t|$ immediately addable edges in $L_t$. These cause the removal of at least $\mu|X_t|/\Delta$ edges from $Y_{t-1}$. Combining with \Cref{lem:analysis-expand-by-delta} or \Cref{lem: approx analysis-expand-by-delta}, which give $|X_t|\ge \frac{\Delta}{10}|Y_{t-1}|$, we have that $|Y_{t-1}|$ decreased by a factor of at least 
        \[  1-(\mu/\Delta)\cdot  \Delta/10 = 1-{1}/{100},
        \]
        which means that $\psi_{t-1,1}$ decreased by at least one.
    \end{itemize}
    Hence, $\Phi$ is lexicographically decreasing. The problem now reduces to bounding the number of distinct viable signatures.
\end{enumerate}

Next, we show that the coordinates of the signature vector are non-decreasing in absolute value at the beginning of each iteration of the main loop. For any $L_i$, the signature vector is non decreasing in absolute value since $|Y_i|\ge (1-\frac1{10})|X_i|$ by \Cref{lem:analysis-same-layer-ratio}. Between any two layers, \Cref{lem:analysis-expand-by-delta} or \Cref{lem: approx analysis-expand-by-delta}, ensures that $|X_i|\ge \frac{\Delta}{10}|Y_{\le i-1}|\ge \frac15|Y_{i-1}|$. Thus, the coordinates of any signature we might encounter are non-decreasing in absolute value.

Finally, we want to upper bound the number of distinct signatures. A \emph{partition} of a positive integer $N$ is a way of writing $N$ as the sum of positive integers without regard to order. The signature can be uniquely represented by a partition of an integer of size at most $O(\ell_{\max}^2+\ell_{\max}\log n)$, as we have proved that the coordinates are non-decreasing and each coordinate is at most $O(\ell_{\max}+\log n)$. As the number of partitions of a positive integer $N$ can be bounded by $\exp(O(\sqrt{N}))$,%
we have an upper bound on the number of distinct signatures of
\[
    2^{\sqrt{O(\ell_{\max}^2+\ell_{\max}\log n)}}.
\]
As we argued above, the number of iterations is upper bounded by the number of distinct signatures, thereby completing the proof.
\end{proof}

The final proposition of this section ensures that the maximum number of layers in an alternating forest throughout either of our algorithm is small.

\begin{proposition}
\label{lem: number of layers under Haxell}
    The maximum number of layers in an alternating forest considered by \textsc{HypergraphMatching} or \textsc{FasterHypergraphMatching} is at most $9\log n/\log \Delta$.
\end{proposition}

\begin{proof}
    By \Cref{lem:analysis-same-layer-ratio} and \Cref{lem:analysis-expand-by-delta} or \Cref{lem: approx analysis-expand-by-delta}, we have $|Y_t|\ge \frac9{10}|X_t|$ and $|X_t| > \frac{\Delta}{10}|Y_{\le t}|$ at the beginning of each main loop.
    Hence, we have $|Y_t| \ge \frac{9\Delta}{100}|Y_{<t}|$. Since all $Y_t$ are disjoint subsets of the partial matchings, we have $|Y_t|\le n$ and hence $t\le \log_{1+\frac{9\Delta}{100}} n \le 9\log n/\log \Delta$. %
\end{proof}

\subsection{Putting Everything Together}
\label{sec: fast matching algos}

\paragraph{Warm Up: A Polynomial-Time Algorithm.} Our first showcase is a polynomial-time hypergraph perfect matching algorithm under the $\varphi$-strong Haxell condition with $\varphi = 4r^2$. Although this result is not directly used in our disjoint paths algorithm, it serves as a useful warm-up.

\begin{theorem}
\label{thm: poly time hypergraph perfect matching}
    Let $H=(A, B, E)$ be an $r$-bounded bipartite hypergraph with totol volume $p$ that satisfies the $\varphi$-strong Haxell condition with $\varphi \ge 4 r^2$, then there exists an algorithm that computes a perfect matching in time $\poly(p)$.
\end{theorem}

\begin{proof}
We run \textsc{HypergraphMatching} with parameter $\Delta=4$.
By \Cref{lem: number of layers under Haxell}, the maximum number of layers $\ell_{\max}\le 5\log n$. By \Cref{lem: upper bound main loop iterations} this implies the number of iterations of the main loop is at most $2^{\sqrt{O(\ell_{\max}^2+\ell_{\max}\log n)}} = \poly(n)$. 
Observe that, each iteration of the \textbf{while} loop in \textsc{CollapseForest} decreases the number of layers by one. Since the number of layers only grows (by one) upon a run of the main loop, we conclude the number of times we perform the steps inside the while loop of \textsc{CollapseForest} is equal to the number of times we execute the main loop, so is also polynomial.
A simple greedy algorithm (see \Cref{prop:greedy half layer}) allows us to run \textsc{BuildLayer} in time $\poly(p)$. Thus, the time performing each line within the loops is also polynomial.
Therefore, we have a polynomial time algorithm for hypergraph perfect matching under $\varphi$-strong Haxell condition.
\end{proof}

\paragraph{Algorithms with Subpolynomial Iterations.}

Finally, we  prove \Cref{thm: quadratic time hypergraph perfect matching,lem: fast perfect matching} based on \textsc{HypergraphMatching} and \textsc{FasterHypergraphMatching}, respectively. We restate them for convenience.

\quadraticLemma*

\begin{proof}[Proof of \Cref{thm: quadratic time hypergraph perfect matching}]
We run the algorithm \textsc{HypergraphMatching} with $\Delta = d(n)$. By \Cref{lem: number of layers under Haxell} the maximum number of layers $\ell_{\max} \le 9\log n/\log d(n)$. Combined with \Cref{lem: upper bound main loop iterations}, this upper bounds the number of iterations of the main loop of \textsc{HypergraphMatching} by 
\[
    2^{\sqrt{O(\ell_{\max}^2+\ell_{\max}\log n)}} \le n^{1/\Omega(\sqrt{\log d(n)})}.
\]    
Since each iteration of the loop in the \Call{CollapseForest}{} procedure decreases $\ell$ by one, the number of swap and superposed-build operations is at most the number of main iterations. The time to perform each build-layer or superpose-build operation is at most $T^*$, and therefore the whole algorithm completes in $O\left(T^* n^{1/\Omega(\sqrt{\log d(n)})}\right)$ time.
\end{proof}

\linearLemma*

\begin{proof}[Proof of \Cref{lem: fast perfect matching}]
We run the algorithm \textsc{FasterHypergraphMatching} with $\Delta = d(n)$.  %
By \Cref{lem: number of layers under Haxell}, the maximum number of layers
$\ell_{\max} \le 9\log n/\log \Delta$. The rest of the proof is the same as in the proof of \Cref{thm: quadratic time hypergraph perfect matching}, except for replacing the running time $T^*$ of build layer operation with $\hat T$ of the $(r',\alpha)$-approximate build layer procedure.
\end{proof}

We note that in both above proofs we are tacitly using the inequalities $\hat T, T^* \ge \Omega (n)$ which guarantee that running build layer and approximate build layer procedures is the most expensive part of the iteration.

\bibliographystyle{alpha}
\bibliography{main}

\appendix 

\section{Omitted Proofs}

\subsection{Proof of \texorpdfstring{\Cref{cor: splitting expanders}}{}}
\label{app: splitting expanders}

We prove the first part of the corollary, and the second part follows exactly the same argument. Let $G_X$ be an $n$-vertex graph with constant conductance and maximum degree at most $9$. Such a graph can be constructed in $O(n)$ time, see \cite[Theorem 2.4]{chuzhoy2020deterministic}. For each edge $(u,v)\in E(G_X)$, we add $k$ copies of $(u,v)$ into the set of demand pairs. As $G_X$ has maximum degree at most $9$, each vertex appears in at most $9k$ pairs. By assumption $\phi^3\delta\ge (73\log n)^3 k\ge (35\log n)^3 \cdot 9k$, so by \Cref{thm: routing conductance} Part I we can compute a set of disjoint paths for the set of demand pairs in time $mn^{1+o(1)}\min\{k, \phi^{-1}\}$.
Let $E_i$ be the union of the disjoint paths of the $i$-th copy of demands for all $(u,v)\in E(G_X)$. We now show that the conductance of each $G_i = (V,E_i)$ is at least $\Omega(\phi^2/\log n)$.

Consider a nonempty subset $S\subsetneq V$ and any $i\in [k]$. Let $\widehat \vol$ be the number of disjoint paths that intersect $S$ but don't have an endpoint in $S$. Let $r = 18 \log (n)/\phi$ be the maximum length of the disjoint paths. We have
\[
    \vol_{G_i}(S)\le (\vol_{G_X}(S)+\widehat\vol)\cdot r
\]
and
\[
    |\partial_{G_i}(S)|\ge |\partial_{G_X}(S)| + \widehat\vol
\]
as the paths of the $(u,v)$-pairs that have $(u,v)\in \partial_{G_X}(S)$ must cross $X,$ and the paths which contribute to $\widehat\vol$ also cross $X$ by definition. Therefore, assuming $\vol_{G_X}(S) \le \vol_{G_X}(V \setminus S)$ we get 
\[
    \frac{|\partial_{G_i}(S)|}{\vol_{G_i}(S)}\ge \frac{|\partial_{G_X}(S)|}{r\cdot \vol_{G_X}(S)} \ge \Omega(\phi/\log n).
\]
On the other hand, if $\vol_{G_X}(S) > \vol_{G_X}(V \setminus S)$ we get
$$\frac{|\partial_{G_i}(S)|}{\vol_{G_i}(V \setminus S)}=\frac{|\partial_{G_i}(V \setminus S)|}{\vol_{G_i}(V \setminus S)}\ge \frac{|\partial_{G_X}(V \setminus S)|}{r\cdot \vol_{G_X}(V \setminus S)} \ge \Omega(\phi/\log n).$$

\subsection{Proof of \texorpdfstring{\Cref{lem:implement half layer Dinitz}}{}}\label{sec: proof dinitz}

    To achieve an $O(mnr)$ running time, the algorithm iterates through every vertex on $G$, and attempts to fulfill all the demands with the same source at a time.  
    Fix a source vertex $s$. 
    Let $\mathcal{T}$ be the set of destinations where $(s, t)$ is a demand pair for all $t\in \mathcal{T}$. The algorithm then adds a super-terminal $t'$, and adds directed edges $(t, t')$ with capacity $\Delta$.

    Intuitively, we now let the algorithm find one \emph{blocking set of $s$-$t'$ shortest paths} at a time, for each shortest distances up to $r$. This is similar to a phase of Dinitz' blocking flow procedure~\cite{Dinitz06}, without any flow augmentation.
    Formally speaking, the algorithm first runs a BFS from $s$ on $G'$, obtaining the distances from $s$ to each vertex on $G'$. All vertices at distance $i$ will be said to be at \emph{level} $i$. Then, the algorithm builds a \emph{level graph} $G_{\rm L}$, which is a directed graph keeping only the edges that goes from level $i$ to level $i+1$ for all $i$. Finally, the algorithm runs a modified DFS on $G_{\rm L}$, obtaining a blocking set of paths from $s$ to $t'$. The rules for the modified DFS are as follows:
    \begin{itemize}
        \item Whenever a new path from $s$ to $t'$ is found with the last edge being $(t, t')$, the algorithm deletes the path from $s$ to $t$ and adds the corresponding hyperedge to $Z'$. The capacity of the edge $(t, t')$ is decremented by $1$. The algorithm then \emph{restarts} the DFS from $s$.
        \item Whenever an edge $(u, v)$ is backtracked because there is no path from $v$ to $t'$ on $G_{\rm L}$. In this case, the edge $(u, v)$ is removed and the search continues from $u$'s other outgoing edge.
    \end{itemize}
    Observe that any blocking set of $s$-$t'$ shortest paths intersect with any other $s$-$t'$ shortest path on $G'$, after repeating this procedure for $r$ times there is no short path of length $r$ fulfilling the demands anymore.
    By iterating through all source vertices, the total running time for this algorithm is then $O(n\times rm) = O(mnr)$ as desired.

\end{document}

%% file: header.tex
\usepackage{fullpage}
\usepackage[utf8]{inputenc}
\usepackage{amsmath,amsfonts,amsthm,amssymb}
\usepackage{setspace}
\usepackage{fancyhdr}
\usepackage{lastpage}
\usepackage{extramarks}
\usepackage{chngpage}
\usepackage{soul,color}
\usepackage{graphicx,float,wrapfig}
\usepackage{listings}
\usepackage{xcolor}      %
\usepackage{enumerate}
\usepackage{xspace}
\usepackage[T1]{fontenc}

\usepackage{xcolor}
\usepackage{nameref}
\definecolor{ForestGreen}{rgb}{0.1333,0.5451,0.1333}
\definecolor{DarkRed}{rgb}{0.65,0,0}
\definecolor{Red}{rgb}{1,0,0}
\usepackage[linktocpage=true,
pagebackref=true,colorlinks,
linkcolor=DarkRed,citecolor=ForestGreen,
bookmarks,bookmarksopen,bookmarksnumbered]
{hyperref}
\usepackage{cleveref}
\usepackage{algpseudocode}
\usepackage{algorithm}  
\usepackage{algorithmicx} 
\usepackage{verbatim}
\usepackage{todonotes}
\usepackage{appendix}

\makeatletter
\def\namedlabel#1#2{\begingroup
    #2%
    \def\@currentlabel{#2}%
    \phantomsection\label{#1}\endgroup
}
\makeatother
\usepackage{thm-restate}
\usepackage{thmtools}

\declaretheorem[numberwithin=section]{theorem}
\declaretheorem[numberlike=theorem]{lemma}

\declaretheorem[numberlike=theorem]{fact}
\declaretheorem[numberlike=theorem]{proposition}

\declaretheorem[numberlike=theorem]{corollary}

\declaretheorem[numberlike=theorem]{problem}

\declaretheorem[numberlike=theorem,style=definition]{definition}

\usepackage{thmtools}
\usepackage{etoolbox} %

\makeatletter
\AtBeginEnvironment{restatelemma}{%
  \let\c@restatelemma\c@lemma
}
\makeatother

\usepackage{mathrsfs}

\newcommand{\R}{\mathbb{R}}
\newcommand{\vol}{\mathrm{Vol}}
\newcommand{\val}{\mathrm{val}}
\newcommand{\leng}{\mathrm{leng}}
\newcommand{\conge}{\mathrm{cong}}
\newcommand{\supp}{\mathrm{supp}}

\newcommand{\poly}{\mathrm{poly}}

\Crefname{ALC@unique}{Line}{Lines}

\ifdefined\ShowComment

\def\thatchaphol#1{\marginpar{$\leftarrow$\fbox{T}}\footnote{$\Rightarrow$~{\sf\textcolor{purple}{#1 --Thatchaphol}}}}

\def\zhongtian#1{\marginpar{$\leftarrow$\fbox{Z}}\footnote{$\Rightarrow$~{\sf\textcolor{blue}{#1 --Zhongtian}}}}

\def\shangen#1{\marginpar{$\leftarrow$\fbox{S}}\footnote{$\Rightarrow$~{\sf\textcolor{orange!50!black}{#1 --shangen}}}}

\def\matija#1{\marginpar{$\leftarrow$\fbox{M}}\footnote{$\Rightarrow$~{\sf\textcolor{green!50!black}{#1 -- matija}}}}

\else
\def\thatchaphol#1{}
\def\zhongtian#1{}
\def\shangen#1{}
\def\matija#1{}

\fi